\documentclass[a4paper,USenglish,cleveref,thm-restate]{lipics-v2021}
\hideLIPIcs
\nolinenumbers
\usepackage{bbding}
\usepackage{soul}
\usepackage{booktabs} 
\usepackage{amsmath,amssymb,amsfonts}%
\usepackage{amsthm}%
\usepackage[ruled,vlined,linesnumbered]{algorithm2e}%
\usepackage{xspace}
\usepackage{thm-restate}
\usepackage{siunitx} 
\usepackage[table]{xcolor}
\usepackage{colortbl} 
\definecolor{lightgray}{RGB}{245,245,245}
\definecolor{lightblue}{RGB}{230,240,255}

\newcommand{\Liin}{\mathsf{L}^{\mathsf{in}}}
\newcommand{\Liout}{\mathsf{L}^{\mathsf{out}}}
\newcommand{\Riin}{\mathsf{R}^{\mathsf{in}}}
\newcommand{\Riout}{\mathsf{R}^{\mathsf{out}}}
\newcommand{\trace}{\mathrm{trace}}

\newcommand{\N}{\ensuremath{\mathbb{N}\xspace}}

\newcommand{\myID}{\ensuremath{\mathsf{ID}\xspace}}

\newcommand{\id}{\mathsf{ID}\xspace}

\newcommand{\LOCAL}{\mathsf{LOCAL}}
\newcommand{\CONGEST}{\mathsf{CONGEST}}

\DeclareMathOperator{\len}{len}

\usepackage[colorinlistoftodos,prependcaption,textsize=tiny]{todonotes}

\title{Efficient Counting and Simulation in Content-Oblivious Rings} 

\titlerunning{Efficient Computing in Content-oblivious Communication} 

\author{Jérémie {Chalopin}}{Aix Marseille Univ, CNRS, LIS, Marseille, France}{jeremie.chalopin@lis-lab.fr}{https://orcid.org/0000-0002-2988-8969}{}

\author{Yi-Jun Chang}{National University of Singapore, Singapore}{cyijun@nus.edu.sg}{https://orcid.org/0000-0002-0109-2432}{}

\author{Giuseppe A. {Di Luna}}{DIAG, Sapienza University of Rome, Italy}{diluna@diag.uniroma1.it}{}{}

\author{Haoran Zhou}{National University of Singapore, Singapore}{haoranz@u.nus.edu}{https://orcid.org/0009-0001-2458-5344}{}

\authorrunning{J. Chalopin, Y.-J. Chang,  G. A. Di Luna, H. Zhou} 

\Copyright{Jérémie Chalopin, Yi-Jun Chang, Giuseppe A.\ Di Luna, and Haoran Zhou} 



\funding{This work has been partially supported by ANR project MIMETIQUE (ANR-25-CE48-4089-01), by the Ministry of Education, Singapore, under its Academic Research Fund Tier 1 (24-1323-A0001), and by the Italian
MUR National Recovery and Resilience Plan funded
by the European Union - NextGenerationEU through
projects SERICS (PE00000014).} 

\EventEditors{}
\EventNoEds{1}
\EventLongTitle{}
\EventShortTitle{}
\EventAcronym{}
\EventYear{}
\EventDate{}
\EventLocation{}
\EventLogo{}
\SeriesVolume{}
\ArticleNo{}

\begin{document}

\maketitle










\begin{abstract}

In the content-oblivious (CO) model (proposed by Censor-Hillel et al.~\cite{censor2023distributed}), processes inhabit an asynchronous network and communicate only by exchanging pulses: zero-size messages that convey no information beyond their existence. 

A series of works has clarified the computational power of this model. In particular, it was shown that, when a leader is present and the network is 2-edge-connected, content-oblivious communication can simulate classical asynchronous message passing. 
Subsequent results extended this equivalence to leaderless oriented and unoriented rings, and, under non-uniform assumptions, to general 2-edge-connected networks.

While these results are decisive from a computability standpoint, they do not address efficiency. The simulator of Censor-Hillel et al. requires $O(n^3b+n^3\log n)$ pulses to emulate the send of a single $b$-bit message, making it impractical even on modest-size networks

In this paper, we therefore focus on message-efficient computation in CO networks. We study the fundamental problem of counting in ring topologies, both because knowing the exact network size is a basic prerequisite for many distributed tasks and because counting immediately implies a broad class of aggregation primitives. We give an algorithm that counts using $O(n^{1.5})$ pulses in anonymous rings with a leader, an $O(n\log^2 n)$ algorithm for counting in rings with IDs. Moreover, we show that any counting algorithm in CO requires $\Omega(n\log n)$ pulses.

Interestingly, in the course of this investigation, we design a simulator for classic message passing that enables parallel neighbor-to-neighbor communication: in one simulated round, each process can send a $b$-bit message to each of its neighbors using only $O(b)$ pulses per process. The simulator extends to general 2-edge-connected networks, after a pre-processing step that requires $O(n^{8}\log n)$ pulses, where $n$ is the number of processes, allowing thus efficient simulation of asynchronous message passing in general 2-edge-connected networks.

\end{abstract}

\thispagestyle{empty}
\newpage
\thispagestyle{empty}
\tableofcontents
\newpage
\pagenumbering{arabic}

\section{Introduction}\label{sec1} 

A key aspect of distributed computing research is the design of computationally efficient primitives under different communication models. A wealth of research has studied how to create efficient algorithms, measured in terms of message complexity, giving us optimal solutions for foundational primitives such as consensus, reliable broadcast, leader election, counting, etc.

A defining aspect of the message-passing models is the amount of information that a message may carry. The usual spectrum in this regard goes from the $\LOCAL$ model, where messages have unrestricted size, to the $\CONGEST$ model, where messages are of logarithmic size, down to models where messages are of constant size.

In this paper, we focus on asynchronous systems where processes exchange messages that carry no content. This is the content-oblivious model (CO model) of Censor-Hillel et al.~\cite{censor2023distributed}. In this model, processes communicate over a bidirectional network, under an asynchronous scheduler, by exchanging pulses. A pulse carries no information, not even the identity of the sending process. This implies that the only information available to a process is the count of the pulses received on each local port. 

A content-oblivious model may arise in the case of fully defective channels, that is, when an adversary is able to corrupt the content of all messages in transit but cannot delete or inject new messages. It is worth noting that there are other realistic physical models in which communication is based on pulses. Examples include ultraviolet communication (where light pulses are used to communicate~\cite{xu2019pulselaser}), neural systems (where electrical spikes across dendrites are used), and asynchronous circuits (where the absence of a clock signal requires pulse-based communication).

Censor-Hillel~et~al.~\cite{censor2023distributed} showed that this weak model is equivalent to asynchronous message passing when a leader is present, and the network is 2-edge connected. Further results~\cite{frei2024content,chalopin_et_al:LIPIcs.DISC.2025.51,cha2025content} gave leader election algorithms for rings and, in the non-uniform case (that is, when an upper bound on the network size is known), for general 2-edge connected networks, showing that in these cases the CO model and classic message passing are computationally equivalent. 

However, computational equivalence says nothing about message efficiency. The simulator of Censor-Hillel et al., although a remarkable landmark from the point of view of computability equivalence, is highly inefficient, generating, in the worst case, $O(n^3b+n^3\log n)$ \cite[Theorem 2]{censor2023distributed} CO messages for each $b$-bit message sent by the simulated algorithm. This high number of messages makes the computation impractical even in networks of modest size. 

In this paper we focus on message-efficient computations for CO networks. We focus in particular on the problem of counting the number of processes in a ring network.
This problem is important per se, as the knowledge of the exact network size is a requirement for many algorithms and computations. Moreover, assuming that each process $p_i$ starts with an input $i_i \in I$, counting the number of occurrences of each input allows us to compute many functions of the processes' inputs, such as $\min$, $\max$, average, median, mode, and similar aggregates. Due to its centrality counting has been investigated in several works \cite{DiLunaViglietta2025EfficientComputation,DiLunaViglietta2022ComputingAnonymousDynamicNetworksLinear,MerrerKermarrecMassoulie2006P2PSizeEstimation} but never in the CO setting.

Quite interestingly, this study of the counting problem leads us to the design of a simulator from a local model to the CO model that is more efficient than the one of Censor-Hillel et al. and works for any 2-edge-connected network.

\subsection{Contributions}

We study distributed computation in oriented rings of $n$ processes communicating via \emph{content-oblivious} messages. Throughout, we assume the presence of a unique leader. In rings where processes have unique IDs, this assumption is without loss of generality, as a content-oblivious leader election algorithm can be executed as a preprocessing step~\cite{chalopin_et_al:LIPIcs.DISC.2025.51,cha2025content,frei2024content}. 

\subparagraph{Self-decomposable aggregation.}
Our primary algorithmic goal is to compute global aggregation functions over data distributed across the ring. We focus on the class of \emph{self-decomposable aggregation functions}, which admit a natural divide-and-combine structure. Formally, an aggregation function $f$ defined over multisets is self-decomposable if there exists an associative operator $\oplus$ such that, for any disjoint multisets $X_1$ and $X_2$,
\[
f(X_1 \uplus X_2) = f(X_1) \oplus f(X_2),
\]
where $\uplus$ denotes the multiset union.

This class captures many fundamental distributed tasks, including counting, summation, maxima, parity, and related semigroup-based computations. Self-decomposability enables hierarchical aggregation: partial results computed on disjoint subsets of processes can be merged.

\subparagraph{Algorithm simulation for content-oblivious rings.}
Our first technical contribution is an efficient algorithm simulator for anonymous content-oblivious rings. In particular, it allows the simulation of a synchronous round of the $\CONGEST[b]$ model, where each process sends a $b$-bit message to each of its neighbors, using only $O(b)$ pulses per process. In this sense, the simulation is asymptotically message-optimal. The simulator supports the emulation of a \emph{logical} ring in which only a subset of processes actively participate in the computation, while the remaining processes act solely as relays that forward information. This flexibility is crucial in our algorithms.

\subparagraph{Efficient aggregation and counting with unique IDs.}
Using this simulator, we obtain an efficient deterministic aggregation algorithm by building a hierarchical decomposition of the ring. This decomposition is computed by simulating a deterministic maximal independent set (MIS) algorithm in the $\LOCAL$ model, which inherently requires unique IDs. The approach yields our main aggregation result.


\begin{restatable}[Efficient aggregation]{theorem}{thmCount}
Let $f$ be a self-decomposable aggregation function defined on multisets over a universe $U$. 
Consider a content-oblivious ring $C$ of $n$ processes with a designated leader, where each process $p$ is equipped with a distinct identifier of at most $\lambda$ bits and holds an element $x_p \in U$.  Assume that for every non-empty subset $S$ of processes,  the bit length of the value $f\left(\biguplus_{p \in S} \{x_p\}\right)$  is at most $\beta$.

Then there exists a deterministic algorithm that quiescently terminates and computes $f\left(\biguplus_{p \in C} \{x_p\}\right)$ at all processes using $O\left((\lambda + \beta) n \log n  \right)$ pulses.
   \label{thm:Count}
\end{restatable}


As a direct corollary, we obtain an efficient deterministic counting algorithm.

\begin{corollary}[Efficient counting]\label{cor:counting}
In a content-oblivious ring of $n$ processes with \emph{distinct} $O(\log n)$-bit IDs and a designated leader, there exists a deterministic, quiescently terminating algorithm that computes $n$ at all processes using $O(n \log^2 n)$ pulses.
\end{corollary}
\begin{proof}
We apply \Cref{thm:Count} with $\lambda = O(\log n)$, aggregation function $f(X) = |X|$, universe $U = \{1\}$, and $x_v = 1$ for all $v \in C$. Then $f\!\left(\biguplus_{v \in C} \{x_v\}\right) = n$, and since $\beta = O(\log n)$, the total number of pulses is $O\!\left(n \log n \, (\lambda + \beta)\right) = O(n \log^2 n)$.
\end{proof}
 
We complement our upper bound with an $\Omega(n \log n)$-message lower bound. This result shows that our  $O(n\log^2 n)$-message counting algorithm is optimal up to a multiplicative $O(\log n)$ factor.

\begin{restatable}[Lower bound]{theorem}{thmlb}
In oriented rings with content-oblivious messages, where processes have unique IDs and there is a pre-selected leader, any algorithm that correctly solves counting when processes do not know $n$ must have an execution with message complexity $\Omega(n \log n)$.
    \label{thm:lb}
\end{restatable}

\subparagraph{Counting with anonymity.}
We next turn to a strictly weaker model: an oriented ring with a unique leader\footnote{We stress that the presence of a predetermined leader is necessary in this setting due to the ring symmetry \cite{angluin1980local}.}, where all processes are \emph{anonymous}. In this setting, deterministic local symmetry-breaking techniques, including the simulation of $\LOCAL$ algorithms that underpin \Cref{thm:Count}, are no longer available.

We show that counting remains possible in this anonymous model, and that it simultaneously enables symmetry breaking. Specifically, our algorithm assigns to each process its clockwise distance from the leader, thereby inducing a unique identifier for every process. As a result, the algorithm not only counts the number of processes but also transforms an anonymous ring with a leader into an ID-equipped ring. This transformation serves as a bootstrapping step that enables the deterministic aggregation and MIS-based algorithms developed earlier.

 \begin{restatable}[Counting with anonymity]{theorem}{countnsqrtn}
There is a deterministic, quiescently terminating algorithm that counts the \ul{size of an anonymous ring} with $O(n^{1.5})$ message complexity, provided that a \ul{unique leader} exists. Moreover, each anonymous process on the ring outputs the \ul{clockwise distance} from the unique leader to itself.
    \label{thm:countnsqrtn}
\end{restatable}

3n
\subparagraph{Finding the minimum and computing the multiset of the input.}

When each process $p$ starts with an input $x_p \in \N$, we also consider two notable aggregation problems: finding the minimum input $x_{\min}$ and computing the multiset of inputs, i.e., the set of distinct inputs together with their multiplicities. For these two problems, we develop ad hoc techniques that improve on our bound of \Cref{thm:Count}. Specifically, we show that:

\begin{restatable}{theorem}{minFinding}
In an anonymous ring of $n$ processes with a designated leader where each process has an input $x_p \in \N$, there exists a deterministic, quiescently terminating algorithm  that computes $x_{\min} = \min_{p} \{x_p\}$ using $O(n|x_{\min}|)$ pulses.
\label{th:minFinding}
\end{restatable}

\begin{restatable}{theorem}{multisetComputation}
In a ring of $n$ processes, with distinct $O(\log n)$-bit IDs and designated leader, where each process has an input $x_p \in \N$, there exists a deterministic, quiescently terminating algorithm  that computes the multiset $\{\{x_p\}\}$ of inputs using $O(nD(B+\log^2 n))$ pulses where $B= \max_p \{|x_p|\}$ and $D$ is the number of different inputs. 
\label{th:multisetComputation}
\end{restatable}

A summary of our results pertaining Counting and related problems is in Table \ref{table:nicetable}. All the algorithms that we  present are {\em composable}, that is they can be executed back-to-back.

\subparagraph{Algorithm simulation for general 2-edge-connected networks.}
Finally, we show that our simulator for content-oblivious rings extends to arbitrary synchronous message-passing algorithms on general 2-edge-connected networks. The main insight is to exploit a classical structural result from graph theory, which guarantees that the edges of any 2-edge-connected graph can be covered by three even-degree subgraphs, each of whose connected components admits an Eulerian cycle~\cite{alon1985covering}.

After a preprocessing phase that computes such a decomposition, we apply our ring simulator along these Eulerian cycles, thereby enabling the simulation of message-passing communication on the original network. As in the ring setting, the resulting simulator is asymptotically optimal: simulating a single round of the $\CONGEST[b]$ model incurs only $O(b)$ pulses of communication per edge.

\begin{restatable}[Algorithm simulation for 2-edge-connected networks]{theorem}{thmSim}
    Let $\mathcal{A}$ be an algorithm in a \ul{message-passing, synchronous, and anonymous} 2-edge-connected network $G = (V,E)$, where each process can send $b$ bits per round on each edge. There is a quiescently terminating \ul{content-oblivious algorithm on a ring with a unique leader} that simulates $\mathcal{A}$, with the following specifications.
    \begin{itemize}
        \item There is a pre-processing step with $O(n^8\log n)$ pulses, where $n$ is the number of processes in $G$.
        \item After that, simulating a round of $\CONGEST[b]$ requires sending $O(b)$ pulses along each edge, i.e, with a constant multiplicative overhead. 
    \end{itemize}
   \label{thm:thmSim}
\end{restatable}

\begin{table}[h!]
\centering
\caption{Summary of the results for the Counting problem and Input multiset problem presented in the paper. $n$ is the number of processes in the system, $B$ indicates the maximum bit len of an input and $D$ the different number of initial inputs.}
\label{table:nicetable}
\rowcolors{2}{lightgray}{white} 
\begin{tabular}{@{}l l l l l@{}}
\toprule
\rowcolor{lightblue}
\textbf{Section} & \textbf{Problem} & \textbf{Leader} & \textbf{IDs} & \textbf{Message Complexity} \\ 
\midrule
\S\ref{sec:nradixn} & Counting & Yes (necessary) & No  & $O(n^{1.5})$ \\

\S\ref{sec:colorcount} & Counting & Yes & Yes & $O(n\cdot \log^{2}n)$ \\

\S \ref{sec:multi} & Inputs Multiset & Yes & Yes   & $O(nD(B+\log^{2} n))$ \\

\S \ref{sec:lb} & Counting & Yes & Yes  & $\Omega(n \cdot \log n)$ \\

\bottomrule
\end{tabular}
\end{table}

\section{Related work}
The investigation of CO models began with Censor-Hillel et al.~\cite{censor2023distributed}, who showed that the presence of a leader and a $2$-edge-connected topology suffices to simulate message-passing algorithms. Their simulator includes an initialization step requiring $O(n^8 \log n)$ pulses, and simulating the transmission of a $b$-bit message needs $O(n^3 b + n^3 \log n)$ pulses.

Other work on the CO model has mainly focused on the \emph{Leader Election} (LE) problem~\cite{frei2024content,cha2025content,chalopin_et_al:LIPIcs.DISC.2025.51}.
Frei et al.~\cite{frei2024content} were the first to investigate LE in rings, giving a terminating LE algorithm for oriented rings that uses $O(n\, \id_{\max})$ pulses, where $\id_{\max}$ denotes the maximum identifier in the network, as well as a stabilizing algorithm for unoriented rings.
This result was extended by Chalopin et al.~\cite{cha2025content}, who presented a terminating LE algorithm for unoriented rings using $O(n\, \id_{\max})$ pulses and, assuming a known upper bound $N$ on the network size, an LE algorithm for general $2$-edge-connected networks using $O(m\, N\, \id_{\min})$ pulses.
The efficiency of LE in non-uniform oriented rings was studied by Chalopin et al.~\cite{chalopin_et_al:LIPIcs.DISC.2025.51}, who gave an algorithm requiring $O(N \, \log \id_{\min})$ pulses.

Censor-Hillel~et~al.~\cite{censor2023distributed} established that no nontrivial task in the CO model can be solved in networks that are not 2-edge-connected. Recently, Chang, Chen, and Zhou~\cite{chang_et_al:LIPIcs.ITCS.2026.36} showed that certain meaningful tasks remain solvable in such networks, without contradicting this impossibility result. In particular, they proved that, given that the network topology is known to all processes, terminating leader election in trees is possible if and only if the tree is not edge-symmetric.

The pursue of message efficiency has been a longstanding overarching research topic in distributed computing. Here we limit to briefly describe the message efficiency concepts in the $\LOCAL$ and $\CONGEST$ model in order to highlight the radical differences with our environment.  

In the $\LOCAL$ model, the relationship between round complexity and message complexity has been investigated in \cite{bitton2019message,dufoulon_et_al:LIPIcs.OPODIS.2024.26}, leading to the conclusion that it is always possible to design algorithms that are simultaneously round-optimal and message-optimal, up to polylogarithmic factors.

In the $\CONGEST$ model, a line of work has developed message-complexity bounds for classic global graph problems, often yielding tight (or near-tight) results. This includes MST   \cite{elkin2020simple,ghaffari_et_al:LIPIcs.DISC.2018.30,king2015construction,pandurangan2019time}, and broader families of optimization problems~\cite{dufoulon_et_al:LIPIcs.ITCS.2024.41}. In addition to these problem-specific bounds, several papers focus explicitly on the relationship between round and message trade-offs~\cite{gmyr_et_al:LIPIcs.DISC.2018.32, haeupler2018round}. 
It is clear that techniques developed in the aformentioned work do not apply to our severely restricted model. 

A closely related strand of work investigates the \emph{beeping} communication model~\cite{ CASTEIGTS201920, CASTEIGTS201932,cornejo2010deploying,VacusZiccardi2025}. In this abstraction, computation proceeds in discrete rounds and, at each round, a process chooses between two actions: it either broadcasts an undifferentiated ``beep'' to its neighborhood or it stays silent and listens. A key feature is that the framework assumes global synchrony. Our model lacks this kind of shared temporal reference, so such timing-based encodings are unavailable. For this reason results in that setting do not extend to our context.

\section{System model}\label{sec:model}
We consider a distributed system of $n$ processes $P = \{p_0, p_1, \ldots, p_{n-1}\}$ that communicate exclusively by message passing. The processes are arranged either in a ring $C$ or in a general 2-edge-connected graph $G = (V,E)$. In a ring $C$, each process has exactly two communication ports, labeled $0$ and~$1$. In a general 2-edge-connected graph $G$, each process corresponds to a vertex in $V$ and has one communication port for each incident edge. Each port connects to exactly one neighboring process and supports bidirectional communication.

The cost of an algorithm is measured by its \emph{message complexity}, defined as the total number of messages sent by all processes during an execution.

\subparagraph{IDs and Anonymity.}
We assume that each process is equipped with a distinct identifier drawn arbitrarily from $\mathbb{N}$. Formally, for every $j \in [0,n-1]$, let $\myID_j$ denote the identifier of process $p_j$. We assume that there exists a distinguished leader $p_{\ell}$.

In some sections we assume an {\em anonymous system} with a distinguished leader, i.e., all processes except the leader start in the same local state.

\subparagraph{Asynchronous System.}
The underlying system is {\em asynchronous}. Messages may experience arbitrary but finite delays, with no known upper bound, and no message is ever dropped. If multiple messages are delivered to a process at the same instant, they are placed into a local buffer from which the process may retrieve them at arbitrary times.

The time needed for a local computation at a process is also arbitrary but bounded, and there is no global clock or shared notion of time. In our correctness and complexity arguments, we adopt the standard abstraction that all non-blocking local actions (i.e., all actions except receiving a message) take zero time. 

\subparagraph{Content-Oblivious Algorithms.}
We focus on {\em content-oblivious algorithms}, in which the payload of messages is irrelevant: the only information conveyed is the bare fact that a message was sent and a process can only see the ports used to receive the message. One can think of each message as an empty string, or equivalently as a {\em pulse}~\cite{frei2024content}.

This abstraction captures an {\em adversarial} environment that may arbitrarily corrupt message contents during transmission. However, the adversary is not allowed to remove messages from the network or to insert extra messages. All processes are fault-free and never crash.  Since messages carry no payload, reordering pulses on a
link cannot affect the behavior of a content-oblivious algorithm. Hence we may
assume FIFO links without loss of generality. 

\subparagraph{Oriented Rings.}
Throughout the paper, we work with {\em oriented} rings. This means that all processes agree on a common direction of {\em clockwise} (CW) versus {\em counter-clockwise} (CCW) around the ring. Concretely, we consider a ring $(p_0, p_1, \ldots, p_{n-1})$ and assume that, for every $j \in [0,n-1]$: at process $p_j$, port~$1$ is connected to $p_{j+1}$, and port~$0$ is connected to $p_{j-1}$, where indices are taken modulo $n$. A message is said to travel clockwise if it is sent on port~$1$ and received on port~$0$ (i.e., it goes from $p_i$ to $p_{i+1}$). Symmetrically, a message travels counter-clockwise if it is sent on port~$0$ and received on port~$1$. The assumption of oriented ring is not restrictive in our model, as the leader can orient the entire ring with a trivial pre-processing step in which a single pulse is sent across the ring.

\subparagraph{Quiescent Termination.}
A \emph{configuration} consists of the local states of all processes together with the multiset of messages currently in transit on communication links. A configuration is \emph{quiescent} if all processes are in a halting state and no messages are in transit. An execution is \emph{quiescently terminating} if it reaches a quiescent configuration after finitely many steps, and an algorithm is \emph{quiescently terminating} if every execution from an initial configuration has this property. A process \emph{quiescently terminates} if, once it enters the halting state, it neither sends nor receives any further messages.




\subparagraph{The LOCAL and CONGEST models.}
In this work, we are interested in simulating classical synchronous message-passing models within an asynchronous content-oblivious framework. In the $\LOCAL$ model~\cite{linial1992locality}, computation proceeds in synchronous rounds, and in each round every process may send a message of unbounded size along each of its incident edges. When communication is restricted to messages of at most $b$ bits per edge per round, the model is referred to as the $\CONGEST[b]$ model~\cite{peleg2000distributed}.


\section{Composability}
\label{sec:composability}

We will often invoke multiple algorithms back-to-back (e.g., run
algorithm~$A$ and, as soon as it returns, start algorithm~$B$).
Since all pulses are indistinguishable, some care is needed in an asynchronous
system: different processes may return from~$A$ at different times, and a process
still executing~$A$ must not confuse pulses from~$B$ with pulses from~$A$.
We handle this using a simple sufficient condition, called the \emph{composable
ending property}:

\begin{definition}[Composable ending property]
\label{def:composable-ending}
Consider a CO algorithm $A$ executed by all processes of an oriented ring with FIFO
links.
We say that $A$ has the \emph{composable ending property} if in every execution:
\begin{itemize}
\item $A$ is quiescently terminating.
\item
There exists a unique \emph{terminator} process that returns
last from~$A$.
\item
There exists a direction $d \in \{0,1\}$ such that $A$ ends with a
\emph{termination wave} initiated by the terminator on port~$d$.
From the moment the terminator sends this pulse until it terminates the execution of~$A$, it never
waits to receive a pulse on port~$d$.
The terminator returns immediately after receiving the termination wave back on
port~$1-d$.
Every other process returns from~$A$ immediately after receiving (on port~$1-d$)
the termination-wave pulse and forwarding exactly one pulse on port~$d$.
\end{itemize}
\end{definition}

\begin{theorem}
\label{thm:compose-any}
If algorithm $A$ has the composable ending property, then it can be composed
with any other algorithm $B$.
\end{theorem}

\begin{proof}
Fix any execution of the composed algorithm (run $A$, then run $B$).
By the quiescent termination condition, once a process starts $B$ it will never
again receive any pulse of~$A$, so $A$-pulses cannot be mistaken as $B$-pulses.

It remains to argue that a process still executing $A$ cannot consume a pulse
sent by a neighbor that has already started $B$.
Let $d$ be the direction given by Definition~\ref{def:composable-ending}.
Consider any process $p$ different than the terminator that is still in~$A$.
Let $p^{-}$ be the predecessor of $p$ along the termination-wave direction, so
that pulses sent by $p^{-}$ on port $d$ are received by $p$ on port~$1-d$.
Process $p^{-}$ starts $B$ only after returning from $A$, hence only after it has
sent the termination-wave pulse to $p$ on this link.
By FIFO, $p$ receives that termination-wave pulse before any later pulse on the
same link, and after receiving it, $p$
returns from~$A$.
Therefore, $p$ cannot consume any $B$-pulse while still executing $A$.

Finally, consider the terminator $t$.
The first process to return from $A$ is the neighbor reached from $t$ by
following port~$d$, and it may start $B$ long before $t$ starts $B$.
Any pulse it sends to $t$ is delivered to $t$ on port~$d$.
By Definition~\ref{def:composable-ending}, after initiating the termination wave
the terminator does not receive on port~$d$, so these pulses cannot be consumed
during $A$ and are buffered until $t$ starts~$B$.
\end{proof}

\section{Counting in anonymous ring with a leader using $O(n^{1.5})$ messages}\label{sec:nradixn}

In this section, we describe an algorithm that counts the size of an anonymous ring, provided the existence of a leader process. Additionally, each process learns its position on the ring eventually: formally, on a ring of size $n$, each process $p$ eventually outputs $n$, the ring size, and an integer $d\in [0,n]$ equal to the clockwise distance from the unique leader to $p$. 

\subparagraph{A naive $O(n^2)$ algorithm} One may immediately come up with the following $O(n^2)$-pulses algorithm, following a ``explore-bounce back'' motif: The leader process sends one pulse clockwise at initiation, and whenever receiving a counter-clockwise pulse. For a non-leader process, it reflects the first clockwise pulse (absorb the clockwise pulse, and send back a counter-clockwise pulse), and behaves as repeater thereafter. The first clockwise traversal indicates finishing counting the ring, and the ring size equals the number of counter-clockwise pulses the leader has received so far. The leader now triggers a counter-clockwise pulse traversal to terminate the counting. The leader can then utilize a simple procedure to broadcast the size count (see \Cref{algo:sndn,algo:rcvn}).

However, it becomes significantly expensive to explore new processes in the later stage of the algorithm. On average, an ``explore-bounce back'' iteration would travel $\Theta(n)$ in distance, hence resulting in the $O(n^2)$ complexity. Our objective in the following presentation is to significantly reduce the maximum distance traveled by ``explore-bounce back'' iteration, by breaking the algorithm into phases.

\countnsqrtn*

The pseudocode procedures for non-leader and leader processes are presented in \Cref{algo:nonleader,algo:leader}, while \Cref{algo:sndn,algo:rcvn} is invoked to communicate the ring size after counting is finished. Upon initiation, the leader process invokes Leader($phase=1$), while non-leader processes invoke NonLeader($phase=1,counted=false,relayed=0$).

\subparagraph{Technical overview} The algorithm progresses in phases, starting from phase 1. In phase $i$, a temporary leader probes $i$ consecutive processes in its \emph{clockwise} neighborhood and transfers temporary leadership to the process with clockwise distance $i$ from it, i.e., the last process probed in this phase. The transfer to leadership is performed by the \emph{counterclockwise} channel, which distinguishes itself from a probing pulse and is visible to all processes on the ring; therefore, all processes progress synchronously in rounds.

\subparagraph{Probing clockwise neighbors in phase $i$} To probe $i$ consecutive clockwise neighbors, a temporary leader $r$ sends and receives $i$ \emph{probing pulses} (Line 3-6 of \Cref{algo:leader}) on its port 1. A process yet to be probed is always executing the Non-leader algorithm (\Cref{algo:nonleader}) and has variable \textit{counted} set to \textit{false}. When a process in the $i$-clockwise neighborhood of $r$ receives a probing pulse from port 0, it reflects the first pulse and sets \textit{counted} to \textit{true} to indicate itself as already probed (Line 3-5 of \Cref{algo:nonleader}). The newly-probed process writes to its global variable $myPhase$ the current phase number; thereafter, it acts as a repeater for any future incoming pulse from port 0, plus the next immediate pulse from port 1 (Line 7-11 of \Cref{algo:nonleader}), while incrementing its variable $relayed$ whenever it relays a round-trip. Therefore, the $j$th clockwise pulse sent by the temporary leader $r$ is reflected by its $j$th clockwise neighbor, before returning to $r$. Upon $r$ finishing receiving $i$ pulses, all processes in $r$'s $i$-clockwise neighborhood have $counted = true$, among which the $i$th clockwise neighbor is the only process with $relayed = 0$. The $relayed$ variable therefore captures the number of processes probed in the same phase as, but after the variable holder.

\subparagraph{Transfering leadership} The temporary leader $r$ is now ready to transfer leadership to its $i$th clockwise neighbor, which we alias as $r'$. First, $r$ sends a pulse on port 0, and receives a pulse on port 1 (Line 12-13 of \Cref{algo:leader}). This counter-clockwise pulse traversal, which we call the \emph{transfer pulse}, triggers Line 12-13 of \Cref{algo:nonleader} at all other processes on the ring. In particular, $r'$ is the only process that enters the branch of Line 14 in \Cref{algo:nonleader}, while all other non-leader processes enter Line 19 of \Cref{algo:nonleader}. At this moment, $r'$ is ready to become the new leader. It sends a pulse on port 1 and receives a pulse on port 0. This clockwise pulse traversal, which we call the \emph{confirmation pulse}, makes all other non-leader processes enter Line 21; therefore, all proceed to the next phase. This confirmation pulse also travels through the previous leader (Line 14-15 of \Cref{algo:leader}), which makes the old leader also proceed to the next phase as a non-leader.

\subparagraph{Detecting the completion of probing all processes} Define the original leader (before any leadership transfer) as $r_1$, similarly, the $i$th temporary leader as $r_i$. Assume now all processes have been probed ($counted = true$), and the temporary leader $r_i$ is about to probe $r_1$ again. This very probing pulse of $r_i$, unlike any previous probing pulses, will travel around the ring in a clockwise direction before arriving at port 0 or $r_i$. To see this, notice every process except $r_i$ triggers Line 7-9 of \Cref{algo:nonleader} as $counted = true$. Hence $r_i$ learns that all processes on the ring have been probed, and performs Line 8-9 of \Cref{algo:leader}, sending 3 pulses traversing the ring in counter-clockwise direction. The first counter-clockwise traversal invokes Line 10-11 of \Cref{algo:nonleader} at all other processes, the second traversal invokes Line 12-13,19 of \Cref{algo:nonleader}, while the third traversal invokes Line 25-26. All processes except $r_i$ are thus informed of the completion of counting and enter RcvCount() to receive the counting result.

\subparagraph{Broadcasting the counting result} As all other processes are aware of the current phase $i$, the last leader $r_i$ only needs to communicate $count$ encoded in binary for all processes to learn the ring size and their respective clockwise distance to the leader. Note that for the communication of $count$ to terminate, $r_i$ needs to send a dedicated $\bot$ symbol at the end of its transmission. Therefore, $r_i$ can use the orientation of 2 consecutive traversals, which has up to 4 possibilities, to encode an alphabet of size 3 ($\{0,1,\bot\}$).

\subparagraph{Ring size} Each process calculate ring size $n = \frac{i(i-1)}{2} + count +1$. The $i-1$ complete phases count $\frac{i(i-1)}{2}$ processes, plus $count$-many successfully reflected probing pulses in phase $i$, plus one - the leader of the $i$-th phase, $r_i$.

\subparagraph{Distance from leader} First, notice that for all processes except the last leader, $relayed-1$-many (for the last leader, $relayed$-many) other processes are counted after itself in the same round. The minus $1$ is necessary due to the counting completion testing: the last leader introduced an additional clockwise + counterclockwise traversal (Line 7-8 of \Cref{algo:leader}), which resulted in an increment of $relayed$ for all processes except the last leader. 

Now, define $k = relayed$ for the last leader, and $k = relayed-1$ for every other process. Depending on the $myPhase$ variable of the process, there are two cases: 

\begin{itemize}
    \item If $myPhase < i$: $d = \frac{myPhase(myPhase+1)}{2} - k -1$. This is because the phases from $1,2,\ldots,myPhase$ are known to have completed. In particular, there are $k$-many processes that are counted after the process of concern in that same iteration, which do not contribute to the clockwise distance starting from the leader. The minus 1 is due to the fact that the leader itself was included in the counting as well.
    
    \item If $myPhase = i$: $d = n-k -1$. In this case, the last phase - $myPhase$ - was not completed. Directly subtracting from the ring size the number of processes probed in the last phase, after probing the process of concern.
\end{itemize}

We dedicate the next lemma to establish the initial condition for every phase $i \leq \lceil{n^{0.5}}\rceil$ as identical to the ending condition of its previous phase. The main purpose it to inductively show that every leadership transfer is successful.

\begin{lemma}
    Consider a ring of processes performing the $O(n^{1.5})$ counting algorithm. Assume at some moment, all processes other than $r_i = p_0$ have started to perform NonLeader($phase = i, counted, relayed$), and $r_i = p_0$ has started to perform Leader($phase = i$).
    Let $(p_0, p_1, \ldots, p_i)$ be a clockwise segment of processes that have $counted = false$.
    
    With a finite delay, the following statements become true.

    \begin{itemize}
        \item[(1)] All processes except $p_0$ and $p_i$ call NonLeader($phase = i+1, counted, relayed$), and their $counted, relayed$ are not changed.
        \item[(2)] Process $p_0$ calls NonLeader($phase = i+1, counted = true, relayed = 0$).
        \item[(3)] Process $p_i$ calls Leader($phase = i+1$) after (1) and (2) become true.
        \item[(4)] Processes $p_1, \ldots,p_i$ have $counted = true$, and $myPhase = i$.
        \item[(5)] Processes $p_j$, $j\leq i$ has $relayed = i-j$.
    \end{itemize}
\end{lemma}

\begin{proof}
    We prove by induction from $i = 1$. 
    
    Recall that phase $i$ has two components: first, the current leader probes $i$ clockwise neighbors, and then transfers leadership to $p_i$, the $i$-th clockwise neighbor. We prove that (4) and (5) is immediately true for every $i$. Given that the $i$ clockwise neighbors of $p_0$ have $counted = false$ initially, they reflect the first probe pulse, and relay the later probing pulses, as described in Line 3-11 of \Cref{algo:nonleader}. Observe that the probing is controlled entirely by $p_0$, and is local to the clockwise segment $(p_0, p_1, \ldots, p_i)$. (4) and (5) are immediate corollaries from the description of Line 3-11 of \Cref{algo:nonleader}.

    Now we prove (1), (2), and (3) first for $i=1$. When so, every process other than $p_0$ and $p_1$ has $counted = false, relayed = 0$, and $myPhase$ undefined. Process $p_0$ transfers leadership to $p_i$ by sending a counterclockwise traversing \emph{transfer pulse}. This pulse is relayed by all processes on the ring by running Line 12-13 of \Cref{algo:nonleader}. The only process that triggers Line 14 of \Cref{algo:nonleader} and becomes a new leader is $p_1$. It sends a clockwise traversing \emph{confirmation pulse}, which eventually arrives at itself. Note that all other processes on the ring see the clockwise traversing \emph{confirmation pulse} only after the counterclockwise traversing \emph{transfer pulse}. At this point, $p_1$ initiates Leader($phase = i+1$). We therefore conclude the proof for (1), (2), and (3) when $i = 1$.

    Now let the lemma hold for $i = k$. For $i = k+1$, notice that we only need to reason that no processes other than $p_i$ can become the new leader, that is, $p_i$ is the only process that can trigger Line 14 of \Cref{algo:nonleader}. For processes that do not belong to the clockwise segment $(p_0, p_1, \ldots, p_i)$, they either have $myPhase$ undefined if they are not probed yet, or $myPhase < i$ if they are probed in an earlier phase. 
\end{proof}

We now prove a lemma stating that every phase $i$, except for the last, incomplete counting phase, is composable: effectively, every process can correctly label a pulse it receives with the phase $i$ that such pulse was produced in.

\begin{lemma}
    For $i \leq \lceil{n^{0.5}}\rceil$, the first $i$ phases can be composed.
\end{lemma}

\begin{proof}
    We show that pulses belonging to consecutive phases are temporally separated. By the end of the $i$th phase, the following condition holds: All processes other than $p_i$ (the last process probed in phase $i$) has started to perform NonLeader($phase = i, counted, relayed$), which is guaranteed by applying clauses (1) and (2) to the $i$th iteration; And that $p_i$ has started to perform Leader($phase = i$), guaranteed by applying clause (3) to the $i$th iteration. If $p_i$ still has $i+1$ consecutive clockwise neighbors $counted = false$, the induction hypothesis for phase $i+1$ is therefore met. We argue that the $i+1$th phase can be correctly concatenated after execution up to phase $i$: by clause (3), the last probed process is the last to quiescently terminate during phase $i$, which also initiates communications of phase $i+1$. Therefore, pulses belonging to phase $i$ and $i+1$ are temporally separated.
\end{proof}

We now prove that the last counting phase has \emph{composable ending property}, hence to facilitate the concatenation of ring size broadcast and receiving afterwards.

\begin{lemma}
    Phase $i = \lceil{n^{0.5}}\rceil$ has composable ending property.
\end{lemma}

\begin{proof}
    Directly from the definition of \Cref{def:composable-ending}. The last counting phase is quiescently terminating, as every process other than last leader terminates upon receiving the last of the three consecutive counter-clockwise traversals initiated by the last leader (Line 8-9 of \Cref{algo:leader}), and is quiescent as the leader also terminates the last after the three consecutive counter-clockwise traversals conclude. The last counter-clockwise pulse traversal is the \emph{termination wave} required by \Cref{def:composable-ending}.
\end{proof}

\begin{theorem}
    The total pulse complexity of the described counting algorithm is $O(n^{1.5})$.
\end{theorem}

\begin{proof}
    For the $i$th phase, $i$ new, unprobed processes are probed and have their $counted$ set to $true$, except for the last phase, after which the entire ring has been probed. Therefore, there are $O(n^{0.5})$ probing phases.
    
    In each phase except the least, the pulses generated came from two sources: The current leader probing $i$ clockwise neighbors, which uses up to $O(i^2) = O(n)$ pulses, and leadership transfer, which comprises two ring traversal pulse waves (a counter-clockwise transfer pulse starting and ending at the old leader, and a clockwise confirmation pulse starting and ending at the next leader), which in total comprise $O(n)$ pulses.
    
    In the last phase, instead of the usual leadership transfer, the last leader produces 4 pulse traversals (first clockwise, next three counter-clockwise) to signal the ending of counting. This is still manageable in $O(n)$ pulses.

    Eventually, the last leader only has to broadcast its internal variable $count$ (defined as the number of processes counted by the last leader in the last phase, therefore $count \leq i$). For each bit of $count$ and additionally an end-of-transmission symbol, two pulse traversals around the ring are required. In total, $O(n\log n)$ pulses are required for broadcasting the ring size.
\end{proof}

\begin{algorithm}[H]
  \DontPrintSemicolon

  \SetKw{Send}{send}
  \SetKw{Receive}{receive}
  \SetKw{Exec}{execute}
  \SetKw{Continue}{continue}

  \caption{NonLeader($phase,counted,relayed$). The algorithm is called at the beginning on each non-leader process with parameters $phase=1, counted=false, relayed=0$, which are also local variables. The variable $myPhase$ is local (to a process) and global (stored across possibly more than one invocation of Nonleader()). }
  \label{algo:nonleader}

  \While{true}{
    \Receive a message on port $q$\;

    \uIf{$\neg counted \land q = 0$}{
    \tcp{first time counted, pong back}
      \Send a message on port $0$\; 
      $counted \leftarrow true$\;
      $myPhase \leftarrow phase$
      \tcp{I was counted in phase $myPhase$}
    }
    \uElseIf{$counted \land q = 0$}{
      $relayed \leftarrow relayed + 1$ \tcp{$relayed$ is used to remember the number of pulses that are counted after me in the same phase.}
      \Send a message on port $1$\;
      \Receive a message on port $1$\;
      \Send a message on port $0$\;
    }
    \ElseIf{$q = 1$}{
      \Send a message on port $0$\;
      \eIf{$ (phase = myPhase) \land (relayed = 0)$}{ \tcp{I have to become the new leader}
        \Send a message on port $1$\; 
        \Receive a message on port $0$\;
        $phase \leftarrow phase + 1$\;
        \Exec Leader$(phase)$\; 
      }{
        \Receive a message on port $q$\;
        \eIf{$q = 0$}{
          \Send a message on port $1$\;
          $phase \leftarrow phase + 1$\;
          \Continue\;
        }{
          \Send a message on port $0$ \tcp{If I receive two ccw messages, the algorithm is done and I have to receive the count.}
          \Exec RcvCount()\;
        }
      }
    }
  }
\end{algorithm}

\begin{algorithm}[H]
  \DontPrintSemicolon
  \footnotesize
  \SetKw{Send}{send}
  \SetKw{Exec}{execute}
  \SetKw{Continue}{continue}
  \SetKw{Break}{break}
  \SetKw{Receive}{receive}

  $count \leftarrow 0$\; 

    \Repeat{$count=phase$}
    {
        \Send a message on port $1$\;
         \Receive a message on port $q$\;
         \eIf{$q=1$}{
        $count \leftarrow count+1$\;
        }{
        \Send \ul{three} messages on port $0$\; 
         \Receive \ul{three} message on port $1$\;  
        \Exec SndCount$(count)$\;
        }
    }
    \Send a message on port $0$\; 
    \Receive a message on port $1$\; 
    \Receive a message on port $0$\;
     \Send a message on port $1$\;
     $phase \leftarrow phase +1$\;
     \Exec NonLeader($phase,true,0$)\;

  \caption{Leader($phase$) called at the beginning on the initial leader process with $phase=1$} 
  \label{algo:leader}
\end{algorithm}

\begin{algorithm}[H]
  \DontPrintSemicolon
  \footnotesize
  \SetKw{Send}{send}
  \SetKw{Exec}{execute}
  \SetKw{Continue}{continue}
  \SetKw{Break}{break}
  \SetKw{Receive}{receive}
\SetKw{Return}{return}

  \For{$i=0$ to $len(count)-1$}{
   \eIf{bit$(count,i)=0$}{
    \Send two messages on port $0$\; 
    \Receive two messages on port $1$\;
   }{
    \Send two messages on port $1$\; 
    \Receive two messages on port $0$\;
   }
  
  }
    \Send a message on port $1$\; 
        \Receive a message on port $0$\;
    \Send a message on port $0$\; 
      \Receive a message on port $1$\;
    \Return $n$\;
  \caption{SndCount($count$)} 
  \label{algo:sndn}
\end{algorithm}

\begin{algorithm}[H]
  \DontPrintSemicolon
  \footnotesize
  \SetKw{Send}{send}
  \SetKw{Exec}{execute}
  \SetKw{Return}{return}
  \SetKw{Break}{break}
  \SetKw{Receive}{receive}
    $L=[]$ \;
    \Repeat{$(q,u) = (0,1)$}
    {
         \Receive a message on port $q$\;
         \Send a message on port $1-q$ \;
         \Receive a message on port $u$\;
          \Send a message on port $1-u$ \;
        \uIf{$(q,u)=(1,1)$}{
        $L.append(0)$\;
        }
        \ElseIf{$(q,u)=(0,0)$}{
         $L.append(1)$\;
        }
    }
 \Return $bitencode(L)$
  \caption{RcvCount()} 
  \label{algo:rcvn}
\end{algorithm}

\section{Simulating the local model: Message exchange on rings}
\label{sec:bitsendingring}

In this section, we describe an algorithm that enables processes to exchange messages of arbitrary size with their neighbors. This communication occurs in parallel across all processes and requires $O(n)$ messages to send and receive a single bit. The algorithm requires the existence of a leader process.

\begin{theorem}
\label{thm:msgxchangering}
Let $\mathcal{A}$ be an algorithm in a \ul{message-passing, synchronous, possibly anonymous} ring,  where each process \emph{may send a different message on each incident edge} in a round, and let $b$ be the maximum number of bits to be sent in the round.
Consider an oriented ring with a unique leader and a designated set of \emph{active} processes (including the leader); all other processes are \emph{relays}.
Let $R$ be the virtual ring obtained by contracting each maximal relay segment into a single link.
There is a quiescently terminating \ul{content-oblivious algorithm on the physical ring} that simulates $\mathcal{A}$ round-by-round on $R$ with multiplicative overhead: each process sends $O(b)$ pulses per simulated round.
\end{theorem}

The pseudocode is shown in Algorithm~\ref{algo:all2allmsg}. This algorithm is executed by all processes. Each process may send one message clockwise and another counterclockwise. Messages are encoded as lists of bits; if a process does not wish to send a message, it may invoke the procedure with an empty list.  
The algorithm performs a loop in which each process sends and receives one bit clockwise and one bit counterclockwise (see lines~\ref{line:all2all:repeat-start}--\ref{line:all2all:increment-phase}). This is done by first sending a bit clockwise (and receiving one counterclockwise), via the {\tt CWBitSending} procedure, and then sending a bit counterclockwise (and receiving one clockwise) via the {\tt CCWBitSending} procedure.  
The loop continues until all processes have finished transmitting their messages. This termination condition is checked by having each process invoke {\tt ComputeOr}, passing a variable set to {\tt true} if it still has bits to send and {\tt false} otherwise. The {\tt ComputeOr} procedure computes the logical OR of these values across all processes.

As we will show, each invoked procedure causes a process to send a constant number of messages, so the overall communication cost per loop iteration is $O(n)$ messages.

In many higher-level algorithms, we have a set of inactive processes that act as a relay by just forwarding messages, these processes are essentially acting as communication links.
We therefore distinguish between \emph{active} processes, which  call \texttt{MsgExchange}, and \emph{relay} processes, which must still forward pulses to allow their neighbors to communicate.
Relay processes execute the procedure \texttt{MsgRelay} (Algorithm~\ref{algo:all2allrelay}).
Intuitively, a relay behaves like a bidirectional link: whenever it receives a pulse on one port, it promptly forwards a pulse on the opposite port, thus preserving the structure of the bit-sending protocol.
However, relays still participate in global termination detection. 
In each iteration, relay processes call \texttt{ComputeOR(false)}.
Since \texttt{ComputeOR} returns the logical OR of these values across all processes, both \texttt{MsgExchange} and \texttt{MsgRelay} loops terminate simultaneously as soon as no active process has bits left to send.
This guarantees that relay processes stop forwarding pulses exactly when the current message-exchange phase is complete.

\begin{algorithm}
  \footnotesize
  \DontPrintSemicolon
  \SetKw{Send}{send}
  \SetKw{Receive}{receive}
  \SetKw{Return}{return}

  \caption{\texttt{MsgExchange}($msg_{cw},msg_{ccw}$): sends message $msg_{cw}$ to your $cw$ neighbour and message $msg_{ccw}$ to your $ccw$ neighbour, while receiving a message from them.}
  \label{algo:all2allmsg}

  $phase=0$\; \label{line:all2all:init-phase}
  $rcvd_{cw}=[]$\; \label{line:all2all:init-rcvdcw}
  $rcvd_{ccw}=[]$\; \label{line:all2all:init-rcvdccw}
  $b_{+}=msg_{cw}[phase]$\tcp{is $\bot$ if $phase>len(msg_{cw})$} \label{line:all2all:get-bplus}
   $b_{-}=msg_{ccw}[phase]$\tcp{is $\bot$ if $phase>len(msg_{ccw})$} \label{line:all2all:get-bminus}
   $active=(b_{+}\neq\bot)\lor(b_{-}\neq\bot)$\; \label{line:all2all:update-active}
  \While{{\tt ComputeOR}($active$)}{ \label{line:all2all:repeat-start}
    $rcvd_{ccw}.append(${\tt CWBitSending}$(b_{+}))$\; \label{line:all2all:send-cw}
    $rcvd_{cw}.append(${\tt CCWBitSending}$(b_{-}))$\; \label{line:all2all:send-ccw}
    $phase++$\; \label{line:all2all:increment-phase}
    $b_{+}=msg_{cw}[phase]$\tcp{is $\bot$ if $phase>len(msg_{cw})$} 
    $b_{-}=msg_{ccw}[phase]$\tcp{is $\bot$ if $phase>len(msg_{ccw})$} 
       $active=(b_{+}\neq\bot)\lor(b_{-}\neq\bot)$\; \label{line:all2all:update-active2}
  }
  \Return bitencode$(rcvd_{cw})$, bitencode$(rcvd_{ccw})$\; \label{line:all2all:return-msgs}
\end{algorithm}

\begin{algorithm}
  \footnotesize
  \DontPrintSemicolon
  \SetKw{Send}{send}
  \SetKw{Receive}{receive}
  \SetKw{Return}{return}

  \caption{\texttt{MsgRelay} .}
  \label{algo:all2allrelay}

  \While{{\tt ComputeOR}($false$)}{ 
    {\tt CWBitRelay}$()$\; 
    {\tt CCWBitRelay}$()$\; 
  }
  \Return\;
\end{algorithm}

\subsection{CCWBitSending Algorithm}

The {\tt CCWBitSending} procedure is shown in Algorithm~\ref{algo:bitsnd}. Each active process $p_i$ invokes the algorithm with a bit $b_i$ that it wishes to send to its counter-clockwise neighbor. The value $b_i$ can be $0$, $1$, or $\bot$, where $\bot$ is a special symbol indicating that no bit is being sent. The algorithm also enables process $p_i$ to receive the bit sent by its clockwise neighbor $p_{i+1}$.

Each process sends messages counter-clockwise through port~$0$, transmitting one, two, or zero messages depending on whether it needs to send a $0$, $1$, or $\bot$, respectively. The process also set a variable $cnt$ to a value equal to the number of counter-clockwise messages it has sent. 

The leader process is the only one that also sends a message clockwise. This is a synchronization message that a process $p_i$ forwards only after receiving confirmation from its counter-clockwise neighbor $p_{i-1}$ that $p_{i-1}$ has received the bit transmitted by $p_i$. 

Once a process receives a message from its clockwise neighbor (i.e., on port~$1$), it responds by sending a clockwise message (see lines~\ref{line:bitsnd:receive} and~\ref{line:bitsnd:forward1}).  

A process then waits until it has received from its counterclockwise neighbor (on port~$0$) a number of messages equal to the number it sent to transmit its bit, plus one (see the predicate of the \texttt{while} loop at line~\ref{line:bitsnd:while-start}). When this occurs, the process knows that the transmitted bit has been successfully received by its counterclockwise neighbor and that one of the received messages is the synchronization one.  

At this point, the process forwards the synchronization message clockwise (line~\ref{line:bitsnd:send-final1}). Note, however, that after forwarding, the process must still acknowledge incoming counterclockwise messages by sending counterclockwise pulses, since its clockwise neighbor may still be transmitting (see the \texttt{repeat--until} loop at line~\ref{line:bitsnd:repeat-start}).

Once the leader receives the synchronization pulse back, it sends a final pulse clockwise (line~\ref{line:bitsnd:send-final1}). This pulse is then forwarded by each process and serves as a signal indicating that they should stop responding to counter-clockwise pulses and that the algorithm has terminated.

\subparagraph{Relay behavior for clockwise bit sending.}
In the description above we focused on \texttt{CWBitSending} as executed by \emph{active}.
However, we may have \emph{inactive} processes that behave like a communication link so that pulses can traverse the ring, but these processes should also terminate exactly when the active \texttt{CCWBitSending} instances terminate.

Such processes execute the relay-side procedure \texttt{CCWBitRelay} (Algorithm~\ref{algo:bitrlay}).
Intuitively, a relay forwards every incoming pulse to the opposite port and keeps track, in a single integer variable $cnt$, of the net number of pulses it has seen:
$cnt$ is initialised to~$0$ and is incremented for each pulse received from port~$1$ and decremented for each pulse received from port~$0$.
Non-leader relays simply forward all pulses and stop when $cnt$ drops below $-1$.
The leader relay additionally injects the initial synchronization pulse on port~$1$, and stops forwarding exactly when it receives the final counterclockwise termination pulse with $cnt=-1$.
At that point $cnt$ becomes $-2$ and the loop condition fails.

This invariant ensures that an interval consisting only of relay processes is indistinguishable from a single link for the active processes executing \texttt{CCWBitSending}, and that all relays terminate synchronously with them at the end of the bit-sending phase.

\subparagraph{\texttt{CWBitSending}} Since the \texttt{CWBitSending} and \texttt{CWBitRelay} algorithms are symmetric to the counter-clockwise ones, their code and correctness proofs are omitted.

\begin{algorithm}
  \footnotesize
  \DontPrintSemicolon
  \SetKw{Send}{send}
  \SetKw{Receive}{receive}
  \SetKw{Return}{return}

  \caption{\texttt{CCWBitSending}($bit$): sends a bit counterclockwise and receives one from the clockwise neighbor. The input $bit \in \{0,1,\bot\}$, where $\bot$ denotes the absence of a bit to be sent. The function returns a value in $\{0,1,\bot\}$.}
  \label{algo:bitsnd}

  \uIf{$bit=0$}{ \label{line:bitsnd:bit0}
    \Send a message on port $0$\; \label{line:bitsnd:send0}
    $cnt \leftarrow 1$\; \label{line:bitsnd:cnt1}
  }\uElseIf{$bit=1$}{ \label{line:bitsnd:bit1}
    \Send two messages on port $0$\; \label{line:bitsnd:send1}
    $cnt \leftarrow 2$\; \label{line:bitsnd:cnt2}
  }\Else{ \label{line:bitsnd:bitbot}
    $cnt \leftarrow 0$\; \label{line:bitsnd:cnt0}
  }

  \If{$leader$}{ \label{line:bitsnd:leader-send}
    \Send a message on port $1$\; \label{line:bitsnd:leader-port1}
  }

  $in \leftarrow 0$\; \label{line:bitsnd:init-in}  

  \While{$cnt \ge 0$}{ \label{line:bitsnd:while-start}
    \Receive a message on port $q$\; \label{line:bitsnd:receive}
    \uIf{$q = 0$}{ \label{line:bitsnd:if-q0}
      $cnt \leftarrow cnt - 1$\; \label{line:bitsnd:dec-cnt}
    }\Else{ \label{line:bitsnd:else-q1}
      \Send a message on port $1$\; \label{line:bitsnd:forward1}
      $in \leftarrow in + 1$\; \label{line:bitsnd:inc-in}
    }
  } \label{line:bitsnd:while-end}

  \Send a message on port $1$\; \label{line:bitsnd:send-final1}

  \uIf{$leader$}{ \label{line:bitsnd:if-leader}
    \Receive a message on port $0$\; \label{line:bitsnd:leader-recv0}
  }\Else{ \label{line:bitsnd:else-nonleader}
    \Repeat{$q = 0$}{ \label{line:bitsnd:repeat-start}
      \Receive a message on port $q$\; \label{line:bitsnd:repeat-recv}
      \uIf{$q = 1$}{ \label{line:bitsnd:repeat-if1}
        $in \leftarrow in + 1$\; 
        \label{line:bitsnd:repeat-incin}
        \Send a message on port $q$\;
        \label{line:bitsnd:sendport1secondcase}
      }\Else{
      \Send a message on port $1 - q$\; \label{line:bitsnd:repeat-send}
      }
    } 
  }

  \uIf{$in = 0$}{ \label{line:bitsnd:return-bot}
    \Return $\bot$\; \label{line:bitsnd:return-bot2}
  }\uElseIf{$in=1$}{ \label{line:bitsnd:return-0}
    \Return $0$\; \label{line:bitsnd:return0}
  }\ElseIf{$in=2$}{ \label{line:bitsnd:return-1}
     \Return $1$\; \label{line:bitsnd:return1}
  }
\end{algorithm}

\begin{algorithm}
  \footnotesize
  \DontPrintSemicolon
  \SetKw{Send}{send}
  \SetKw{Receive}{receive}
  \SetKw{Return}{return}

  \caption{\texttt{CCWBitRelay}(): participate to the {\tt CCWBitSending} as a relay process, a relay will behave as a link but it will terminate its execution synchronously with the active processes executing {\tt CCWBitSending}.}
  \label{algo:bitrlay}
   $cnt \leftarrow 0$\; 
    \If{$leader$}{ \label{line:bitrlay:leader-send}
    \Send a message on port $1$\; \label{line:bitrlay:leader-port1}
  }

   \label{line:bitrlay:cnt0}
  \While{$cnt \ge -1$}{ \label{line:bitrlay:while-start}
    \Receive a message on port $q$\; \label{line:bitrlay:receive}
    \If{$\neg (leader \land q=0 \land cnt=-1)$}{
    \Send a message on port $1-q$\;
    }
    \uIf{$q = 1$}{ \label{line:bitrlay:if-q0}
      $cnt \leftarrow cnt + 1$\; \label{line:bitrlay:dec-cnt}
    }\Else{ \label{line:bitrlay:else-q1}
      $cnt \leftarrow cnt -1$\; \label{line:bitrlay:inc-in}
    }
  } \label{line:bitrlay:while-end}
   
\end{algorithm}

\subsubsection{Correctness of CCWBitsending}

We now prove correctness of \texttt{CCWBitSending} and
\texttt{CCWBitRelay}.  

\begin{lemma}\label{lem:ccw-active-ring}
Assume that all processes execute \texttt{CCWBitSending} (no relays).
Then every process terminates and each $p_i$ returns $b_{i+1}$, the bit of its
clockwise neighbor. 
During the execution of the algorithm at most $6n$ pulses are sent. 
Moreover, \texttt{CCWBitSending} has the composable ending property.
\end{lemma}
\begin{proof}
Let $\kappa(\bot)=0$, $\kappa(0)=1$, $\kappa(1)=2$.
By lines~\ref{line:bitsnd:bit0}--\ref{line:bitsnd:cnt0}, $p_i$ sends
$\kappa(b_i)$ messages on port~$0$ and sets $cnt_i \gets \kappa(b_i)$.

The only sends on port~$0$ in Algorithm~\ref{algo:bitsnd} are the initial
bit-encoding sends in lines~\ref{line:bitsnd:send0}--\ref{line:bitsnd:send1}.
All synchronization and termination pulses are sent on port~$1$.

Hence every message received on port~$1$ by $p_i$ is a bit-encoding pulse
originally sent on port~$0$ by its clockwise neighbor $p_{i+1}$.
There are exactly $\kappa(b_{i+1})$ such pulses in the whole execution.
Moreover, $in_i$ is incremented \emph{only} when $p_i$ receives on port~$1$
(lines~\ref{line:bitsnd:inc-in} and~\ref{line:bitsnd:repeat-incin}), and it is
incremented exactly once per such receive.
Thus, if we show that $p_i$ is still executing the algorithm when each of
those $\kappa(b_{i+1})$ pulses arrives, then necessarily $in_i = \kappa(b_{i+1})$.

The leader sends one extra synchronization pulse clockwise on port~$1$ at
line~\ref{line:bitsnd:leader-port1}, this pulse will do two loops around the ring, at the end of the first loop the leader will know that each process has received the bit sent to him (this is the transmission end wave), the second loop is used to make processes terminate (this is the termination wave).

Now look at $p_i$ and its counterclockwise neighbor $p_{i-1}$.
While $p_i$ is in the \texttt{while}-loop, $cnt_i$ decreases only on
receives from port~$0$, i.e., on pulses sent on port~$1$ by $p_{i-1}$.
Symmetrically, $p_{i-1}$ decreases its own counter $cnt_{i-1}$ only on
receives from its port~$0$, i.e., on pulses sent on port~$1$ by $p_{i-2}$,
and so on around the ring.

Crucially, $p_i$ leaves the \texttt{while}-loop only once $cnt_i < 0$
(line~\ref{line:bitsnd:while-start}), i.e., after it has received strictly
more than $\kappa(b_i)$ pulses from $p_{i-1}$ on port~$0$.  
The first $\kappa(b_i)$ of these pulses are 
acknowledgments that $p_{i-1}$ has already seen all $\kappa(b_i)$ bit pulses
sent by $p_i$ on port~$0$ (recall that $p_{i-1}$ always reply to a pulse on port $1$ either by line \ref{line:bitsnd:forward1} or by line \ref{line:bitsnd:sendport1secondcase}, we are for now assuming that $p_{i-1}$ does not terminate); the “$+1$” comes from the  wave initiated
by the leader’s termination end pulse and propagated via forwards on
port~$1$ (at line \ref{line:bitsnd:send-final1}). This can be shown by a simple induction on processes.
Therefore $p_i$ cannot send its transmission end pulse
(line~\ref{line:bitsnd:send-final1}) before $p_{i-1}$ has finished processing
all of $p_i$’s bit pulses, this also implies that $p_{i-1}$ cannot terminate before $p_{i}$ executes line~\ref{line:bitsnd:send-final1}, remember that the termination pulse is the second loop of the transmission end pulse.

Applying this argument consistently around the ring, we obtain that the
leader’s transmission end pulse and the local $cnt$-tests ensure no process forwards the final clockwise transmission end wave
until its counterclockwise neighbor has consumed all relevant bit pulses.
In particular, once the transmission end pulse reaches again the leader (that then exits the while at line \ref{line:bitsnd:while-start} and sends the termination wave at line \ref{line:bitsnd:send-final1}) all transmissions have been received.

Consequently, for each $i$ all $\kappa(b_{i+1})$ bit pulses sent by
$p_{i+1}$ on port~$0$ are received on port~$1$ at $p_i$ while $p_i$ is still
executing; none can be lost or delayed past $p_i$’s termination.
Thus $in_i$ is incremented exactly $\kappa(b_{i+1})$ times, and we conclude
that $in_i = \kappa(b_{i+1})$.

Notice that once the leader sends the termination wave (executing line \ref{line:bitsnd:send-final1}) all other processes are executing 
the repeat-until loop at line \ref{line:bitsnd:repeat-start} and no other message is in the network, thus each process receiving it will forward the pulse (at line \ref{line:bitsnd:repeat-send}) and terminate with the correct output. The leader will be the last process to receive the pulse back and it will be the last to terminate with a correct output leaving no remaining pulses in the network. Thus \texttt{CCWBitSending} satisfies Definition~\ref{def:composable-ending} with terminator = leader and $d=1$.

Finally, from the above a simple counting on the number of pulse sent by each process show that this is at most $6$ (two pulses to transmit its bit, two pulses to acknowledge the received bits, and two pulses one for the transmission end and one for the termination wave), thus a total bound of $6n$ pulses generated by the algorithm. 
\end{proof}

\begin{lemma}\label{lem:ccw-relay}
Consider a maximal interval of processes that execute
\texttt{CCWBitRelay} (Algorithm~\ref{algo:bitrlay}), while all other
processes execute \texttt{CCWBitSending}.
Then this interval is equivalent to a single link between its
two neighboring active processes and all relays in the interval terminate. If the leader is a \texttt{CCWBitRelay} it will be the last process to terminate.
\end{lemma}

\begin{proof}
Each relay forwards every incoming pulse on port~$q$ to port~$1-q$
(except in the single special case described below) and updates
$cnt$ by $+1$ on port~$1$ and by $-1$ on port~$0$
(lines~\ref{line:bitrlay:receive}--\ref{line:bitrlay:inc-in}).
Thus pulses traverse the relay interval without duplication or loss, and in
the same direction as in a single-link segment.

Non-leader relays run the loop while $cnt \ge -1$, so they stop only after
the net number of pulses from port~$0$ exceeds by at least~$1$ those from
port~$1$, which happens only when a relay forwards the termination pulse (observe that the $cnt=-1$ when the relay forwards the transmission end pulse).

The leader relay additionally injects one synchronization pulse on port~$1$
(line~\ref{line:bitrlay:leader-port1}) and suppresses forwarding of exactly
one final pulse arriving from port~$0$ when $cnt=-1$ (guard in the
\texttt{if}-statement), this pulse is the termination pulse that is reaching back the leader.
Processing that last pulse still decrements $cnt$ to $-2$, so the loop
condition fails and the leader relay terminates.
From the point of view of its two active neighbors, thus every pulse entering the
interval eventually exits on the other side and each relay terminates exactly after forwarding the termination pulse.
\end{proof}

\begin{theorem}\label{thm:ccw-correct}
In an oriented ring where active processes run \texttt{CCWBitSending} with
inputs $b_i \in \{0,1,\bot\}$ and all other processes run
\texttt{CCWBitRelay}, every process terminates and each active process $p_i$
returns exactly the bit of its clockwise neighbor $p_{i+1}$ (treating relay
segments as links). During the execution of the algorithm at most $6n$ pulses are sent.  Moreover, the combined execution of
\texttt{CCWBitSending}/\texttt{CCWBitRelay} has the composable ending
property.

\end{theorem}

\begin{proof}
Collapse each maximal relay interval into a single abstract link.
By Lemma~\ref{lem:ccw-relay}, this does not change the behavior observed by
the active endpoints, and all relays terminate when the endpoints do.
On the resulting virtual ring, all processes are active and execute
\texttt{CCWBitSending}, so Lemma~\ref{lem:ccw-active-ring} applies: every
process terminates and $p_i$ returns $b_{i+1}$.
The composable ending property follows from	Lemma~\ref{lem:ccw-active-ring} on the virtual ring, together with the fact
	that relays only forward pulses and terminate in the same final termination
	wave (Lemma~\ref{lem:ccw-relay}).
\end{proof}

\subsection{ComputeOR Algorithm}

The pseudocode is shown in Algorithm~\ref{algo:computeor_simplified}. In this algorithm, each process $p_i$ takes as input a Boolean variable $bool_i$, and the algorithm allows all processes to compute the global logical OR $\bigvee_{p_i \in \Pi} bool_i$.

The leader process $p_\ell$ is the one that kick-starts the process by sending a counter-clockwise message if $bool_l$ is false and clockwise otherwise.

Suppose that $bool_l$ is \texttt{false}. The counter-clockwise message sent by $p_\ell$ will be forwarded by each non-leader process $p_i$ for which $bool_i = \texttt{false}$. If all processes have $bool_i = \texttt{false}$, the leader will eventually receive the message back on port~1 (see line~\ref{line:computeor:leader-q1}). The leader will then send another counter-clockwise message, which will circulate around the ring and inform each process that the result of the OR operation is \texttt{false}. Each process detects this condition upon receiving two consecutive messages on port~1 (see line~\ref{line:computeor:nonleader-false-if}), after which it decides that the output is \texttt{false}. When this final message returns to the leader, the leader also sets its output to \texttt{false}.

Suppose now that at least one process $p_j$ has $bool_j = \texttt{true}$, restricting ourselves now to the case where $p_j \neq p_{\ell}$. In this case, when $p_j$ receives the counter-clockwise message initiated by the leader, it will send a clockwise message instead of forwarding the counter-clockwise one (see lines~\ref{line:computeor:nonleader-true-send1}--\ref{line:computeor:nonleader-true-send0}). This clockwise message propagates back toward the leader, indicating to each process receiving this message that at least one process holds a \texttt{true} value. All processes in the clockwise segment between $p_\ell$ and $p_j$ (excluding $p_\ell$ and $p_j$) terminate outputting true when the message is received (they received two consecutive messages on ports that are not both $1$). When the leader receives this message on port~0 (see lines~\ref{line:computeor:leader-q0-booltrue}--\ref{line:computeor:leader-q0-recv1}), it concludes that the global OR is \texttt{true}. The leader then sends a message clockwise. This message serves to inform all processes in the clockwise segment from $p_\ell$ to $p_j$ that the result of the OR is \texttt{true}. The message is forwarded until it reaches $p_j$, which is waiting for it at line~\ref{line:computeor:nonleader-true-recv}. At this point, $p_j$ sends the message back counterclockwise. This returning message ensures that all processes in the aforementioned segment with input \texttt{false} can decide: they will detect two consecutive messages arriving on different ports (not both on port~1) and thus conclude with output \texttt{true}. Processes with input \texttt{true} also terminate during this phase, as they execute lines~\ref{line:computeor:nonleader-true-send1}--\ref{line:computeor:nonleader-true-send0}. Finally, the leader decides that the output is \texttt{true} upon receiving this message back.

When $bool_l = \texttt{true}$, the leader sends a first message clockwise. This message is forwarded by all processes in the network until it returns to the leader. Any forwarding process with input \texttt{true} then waits for pulse arriving from the counter-clockwise direction. For this reason, the leader subsequently sends a counter-clockwise message to make all processes commit to a value. This counter-clockwise message is forwarded by all processes, which commit upon its reception to output $\texttt{true}$. Processes with output $\texttt{false}$ recognize this condition by receiving two consecutive messages that do not both arrive on port~1.

Finally, to make \texttt{ComputeOR} safely composable with subsequent procedures
(Section~\ref{sec:composability}), we append a last \emph{clockwise barrier wave}
after the OR value has been determined. This barrier is collectively executed before the termination of \texttt{ComputeOR}.
The leader sends one final pulse clockwise; each non-leader forwards it once and
then returns; finally, the leader receives the pulse back and returns.
This adds exactly $n$ pulses and ensures that the procedure has the composable
ending property.

\begin{algorithm}
  \footnotesize
  \DontPrintSemicolon
  \SetKw{Send}{send}
  \SetKw{Receive}{receive}
  \SetKw{Return}{return}

  \caption{\texttt{ComputeOR}($bool$): returns $true$ if at least one process called it with $true$, $false$ otherwise.}
  \label{algo:computeor_simplified}

  $ans \leftarrow true$\;

  \If{$leader$}{ \label{line:computeor:if-leader}
    \uIf{$bool$}{ \label{line:computeor:leader-booltrue}
      \Send a message on port $1$\; \label{line:computeor:leader-send1}
    }\Else{ \label{line:computeor:leader-boolfalse}
      \Send a message on port $0$\; \label{line:computeor:leader-send0}
    }
  }

  \Receive message on port $q$\; \label{line:computeor:receive-q}

  \uIf{$leader$}{ \label{line:computeor:leader-block}
    \uIf{$q = 1$}{ \label{line:computeor:leader-q1}
      \Send message on port $0$\; \label{line:computeor:leader-q1-send0}
      \Receive message on port $1$\; \label{line:computeor:leader-q1-recv1}
      $ans \leftarrow false$\; \label{line:computeor:leader-return-false1}
    }
    \Else{ 
      \uIf{$bool$}{ \label{line:computeor:leader-q0-booltrue}
        \Send message on port $0$\; \label{line:computeor:leader-q0-send0}
      }\Else{ \label{line:computeor:leader-q0-boolfalse}
        \Send message on port $1$\; \label{line:computeor:leader-q0-send1}
      }
      \Receive message on port $1$\; \label{line:computeor:leader-q0-recv1}
    }
  }
  \Else{ \label{line:computeor:nonleader}
    \uIf{$bool$}{ \label{line:computeor:nonleader-true-start}
      \Send message on port $1$\; \label{line:computeor:nonleader-true-send1}
      \Receive message on port $1 - q$\; \label{line:computeor:nonleader-true-recv}
      \Send message on port $0$\; \label{line:computeor:nonleader-true-send0}
    }\Else{ \label{line:computeor:nonleader-false-start}
      \Send message on port $1 - q$\; \label{line:computeor:nonleader-false-send1q}
      \Receive message on port $u$\; \label{line:computeor:nonleader-false-recvu}
      \Send message on port $1 - u$\; \label{line:computeor:nonleader-false-send1u}
      \If{$u = 1 \land q = 1$}{ \label{line:computeor:nonleader-false-if}
        $ans \leftarrow false$\; \label{line:computeor:nonleader-return-false}
      }
    } \label{line:computeor:nonleader-false-end}
  }

  \tcp{Barrier for composability} 
  \uIf{$leader$}{
    \Send message on port $1$\;
    \Receive message on port $0$\;
  }\Else{
    \Receive message on port $0$\;
    \Send message on port $1$\;
  }

  \Return $ans$\; \label{line:computeor:return-true}
\end{algorithm}

\subsubsection{Correctness of the \texttt{ComputeOR} Algorithm}

We now prove the correctness of Algorithm~\ref{algo:computeor_simplified}.

To show that the returned values are correct, we distinguish cases based on
the vector $(\mathit{bool}_i)_{p_i \in \Pi}$. We want also the remark that, by construction, at each instant of execution of \texttt{ComputeOR} at most one pulse is present in the network.

\begin{lemma}\label{lem:all-false}
Assume $bool_i = \texttt{false}$ for every process $p_i$. Then every process
terminates and outputs $\texttt{false}$. Moreover, this execution has the composable ending property.
\end{lemma}

\begin{proof}
Since $bool_\ell = \texttt{false}$, the leader $p_\ell$ sends its initial
pulse counter-clockwise on port~$0$ at
line~\ref{line:computeor:leader-send0}.
Consider the pulse that starts at $p_\ell$ on port~$0$. 

Let $p_i$ be the first non-leader reached counter-clockwise from $p_\ell$.
By orientation, this pulse is received on port~$1$, hence in
line~\ref{line:computeor:receive-q} we have $q_i = 1$. Since $bool_i =
\texttt{false}$, $p_i$ executes the ``false'' branch at
line~\ref{line:computeor:nonleader-false-start} and sends a message on port
$1-q_i = 0$ at line~\ref{line:computeor:nonleader-false-send1q}, i.e., further
counter-clockwise.

By induction along the counter-clockwise direction, the same holds at every
non-leader: the first pulse received is this pulse (arriving on port~$1$),
so $q_i = 1$, and the process forwards the pulse on port~$0$. Eventually the
pulse reaches $p_\ell$ from its clockwise neighbor and is received on
port~$1$, so at the leader we have $q_\ell = 1$ in
line~\ref{line:computeor:receive-q}.

	Since $q_\ell = 1$, the leader takes the branch at
	line~\ref{line:computeor:leader-q1}. It sends another pulse on port~$0$ at
	line~\ref{line:computeor:leader-q1-send0}, starting a \emph{second}
	counter-clockwise wave, and then waits to receive a pulse on port~$1$ at
	line~\ref{line:computeor:leader-q1-recv1}. Upon that receive, it sets
	$ans \leftarrow \texttt{false}$ at line~\ref{line:computeor:leader-return-false1}.

Now consider the second pulse emitted by the leader on port~$0$. Each non-leader $p_i$ has already executed
lines~\ref{line:computeor:receive-q}
and~\ref{line:computeor:nonleader-false-send1q}, and is blocked at its second
receive in line~\ref{line:computeor:nonleader-false-recvu}. The next message
it receives is exactly this second wave, coming from its clockwise neighbor
on port~$1$. Hence $u_i = 1$ for all non-leaders.
Process $p_i$ then sends on port $1-u_i = 0$ at
	line~\ref{line:computeor:nonleader-false-send1u}, forwarding the second wave,
	and since $u_i = 1 \land q_i = 1$ holds at
	line~\ref{line:computeor:nonleader-false-if}, it sets
	$ans \leftarrow \texttt{false}$ at line~\ref{line:computeor:nonleader-return-false}.

	The second wave eventually reaches the leader from its clockwise neighbor on
	port~$1$, unlocking the receive at
	line~\ref{line:computeor:leader-q1-recv1}, after which the leader has
	$ans=\texttt{false}$.

	At this point the ring is quiescent, and the algorithm executes the final
	clockwise barrier wave (lines~\ref{line:computeor:return-true} and above).
	Each process returns $\texttt{false}$ after forwarding the barrier, and the
	leader returns last after receiving it back.

\end{proof}

\begin{lemma}\label{lem:leader-false-some-true}
Assume $bool_\ell = \texttt{false}$ and there exists at least one non-leader
$p_j$ with $bool_j = \texttt{true}$. Then every process terminates and
outputs $\texttt{true}$. Morevoer, this execution has the composable ending property.
\end{lemma}

\begin{proof}
Let $p_j$ be the first process with input $\texttt{true}$ encountered when
moving counter-clockwise from the leader. That is, $p_j$ is the first $\texttt{true}$ process on
the sequence $p_\ell, p_{\ell-1}, p_{\ell-2},\dots$.

Partition the ring into three disjoint sets:
\begin{itemize}
  \item $A$: the non-leaders between $p_\ell$ and $p_j$ along this
        counter-clockwise path (excluding $p_\ell$ and $p_j$);
  \item $B$: the non-leaders on the other (clockwise) arc between $p_\ell$
        and $p_j$;
  \item the two distinguished processes $p_\ell$ and $p_j$.
\end{itemize}

Since $bool_\ell = \texttt{false}$, $p_\ell$ sends a pulse on port~$0$ at
line~\ref{line:computeor:leader-send0}. For each $p_i \in A$, the first
pulse received (at line~\ref{line:computeor:receive-q}) is this,
arriving on port~$1$, so $q_i = 1$. Because $bool_i = \texttt{false}$, each such $p_i$ executes the false branch at
line~\ref{line:computeor:nonleader-false-start}, sends on port $1-q_i = 0$
and waits at line~\ref{line:computeor:nonleader-false-recvu}. Hence the pulse
is forwarded counter-clockwise through $A$ and reaches $p_j$ from its
clockwise neighbor on port~$1$, giving $q_j = 1$.

	Since $bool_j = \texttt{true}$ and $q_j = 1$, $p_j$ executes the true branch
	at line~\ref{line:computeor:nonleader-true-start}. It sends a message on
	port~$1$ (clockwise) at line~\ref{line:computeor:nonleader-true-send1} and
		then waits at line~\ref{line:computeor:nonleader-true-recv} for a message on
		port $1-q_j = 0$. It will later send on port~$0$ at
		line~\ref{line:computeor:nonleader-true-send0} and then return $\texttt{true}$
		after the final barrier wave.

Consider the pulse that $p_j$ sends on port~$1$. It travels clockwise
through all processes in $A$ and eventually reaches the leader.

Each $p_i \in A$ has already $q_i = 1$, sent a counter-clockwise message on
port~$0$, and is blocked at its second receive in
line~\ref{line:computeor:nonleader-false-recvu}. The next message delivered
to $p_i$ is this clockwise wave, arriving from its counter-clockwise neighbor
on port~$0$, so $u_i = 0$. Process $p_i$ then sends on port $1-u_i = 1$ at
	line~\ref{line:computeor:nonleader-false-send1u}, forwarding the back wave,
	and since $u_i = 0$, the condition $u_i = 1 \land q_i = 1$ at
	line~\ref{line:computeor:nonleader-false-if} is false. Hence $p_i$ keeps
	$ans=\texttt{true}$.

The same clockwise wave eventually reaches the leader, which receives it
as its first message at line~\ref{line:computeor:receive-q} on port~$0$,
hence $q_\ell = 0$.

At the leader we have $q_\ell = 0$ and $bool_\ell = \texttt{false}$. So it
takes the ``else'' branch of line~\ref{line:computeor:leader-q1} and the
``else'' at line~\ref{line:computeor:leader-q0-boolfalse}, sending a message
on port~$1$ at line~\ref{line:computeor:leader-q0-send1}, and then waiting for a message on port~$1$ at
line~\ref{line:computeor:leader-q0-recv1}.

Each $p_i \in B$ has not yet received any message. Its first receive at
line~\ref{line:computeor:receive-q} is the pulse, arriving from its
counter-clockwise neighbor on port~$0$, so $q_i = 0$. Then:
\begin{itemize}
  \item if $bool_i = \texttt{false}$, $p_i$ executes the false branch, sends
        on port $1-q_i = 1$ and waits for its second message at
        line~\ref{line:computeor:nonleader-false-recvu};
  \item if $bool_i = \texttt{true}$, $p_i$ executes the true branch, sends on
        port~$1$ and waits for its second message on port $1-q_i = 1$ at
        line~\ref{line:computeor:nonleader-true-recv}.
\end{itemize}
In either case, the pulse continues clockwise through $B$ until it
reaches $p_j$ from the other side on port~$0$.
Recall that $p_j$ is blocked waiting at
line~\ref{line:computeor:nonleader-true-recv} for a message on port $1-q_j =
0$. The pulse now arrives at $p_j$ from its clockwise neighbor on
	port~$0$, satisfying this receive. Then $p_j$ sends a message on port~$0$
	(counter-clockwise) at line~\ref{line:computeor:nonleader-true-send0} and
	keeps $ans=\texttt{true}$.
The counter-clockwise message sent by $p_j$ travels from $p_j$ back to the
leader through the processes in $B$. It is sent on port~$0$ and thus each
predecessor in the counter-clockwise direction receives it on port~$1$.

For a process $p_i \in B$:
\begin{itemize}
  \item If $bool_i = \texttt{false}$, then is blocked at
        line~\ref{line:computeor:nonleader-false-recvu}. It now receives the
        pulse on port~$1$, so $u_i = 1$. It sends on port $1-u_i =
        0$ at line~\ref{line:computeor:nonleader-false-send1u}, forwarding
	        the pulse, and since $q_i = 0$ the condition
	        $u_i = 1 \land q_i = 1$ is false; thus it keeps $ans=\texttt{true}$.
	  \item If $bool_i = \texttt{true}$, then is blocked at
	        line~\ref{line:computeor:nonleader-true-recv} waiting on port $1-q_i
	        = 1$. The termination wave arrives on port~$1$ and satisfies this
	        receive. The process then sends on port~$0$ and keeps
	        $ans=\texttt{true}$.
\end{itemize}
Hence every process in $B$ has $ans=\texttt{true}$.

	Finally, the pulse reaches the leader from its clockwise neighbor
	on port~$1$. This is the message the leader is waiting for at
	line~\ref{line:computeor:leader-q0-recv1}. At this point the ring is quiescent
	and all processes execute the final clockwise barrier wave, after which every
	process returns $\texttt{true}$ (line~\ref{line:computeor:return-true}), with the
	leader returning last.
\end{proof}

\begin{lemma}\label{lem:leader-true}
Assume $bool_\ell = \texttt{true}$ (other inputs arbitrary). Then every
process terminates and outputs $\texttt{true}$. Morever, the execution has the composable ending property.
\end{lemma}

\begin{proof}

Since $bool_\ell = \texttt{true}$, the leader sends its initial pulse
clockwise on port~$1$ at line~\ref{line:computeor:leader-send1}. 

Follow this pulse from $p_\ell$. At each non-leader $p_i$, this wave
is received from its counter-clockwise neighbor on port~$0$, hence the first
receive in line~\ref{line:computeor:receive-q} sets $q_i = 0$.

If $bool_i = \texttt{false}$, then $p_i$ executes the false branch at
line~\ref{line:computeor:nonleader-false-start}, sends on port $1-q_i = 1$
and waits at line~\ref{line:computeor:nonleader-false-recvu}. If $bool_i =
\texttt{true}$, it executes the true branch at
line~\ref{line:computeor:nonleader-true-start}, sends on port~$1$ and waits
at line~\ref{line:computeor:nonleader-true-recv} on port $1-q_i = 1$.

In all cases, the pulse is forwarded clockwise on port~$1$, so it
does a full tour of the ring and eventually returns to the leader from its
counter-clockwise neighbor on port~$0$. Therefore, the leader's first
receive at line~\ref{line:computeor:receive-q} has $q_\ell = 0$.

At the leader, we have $q_\ell = 0$ and $bool_\ell = \texttt{true}$. Thus it
executes the ``if'' branch at
line~\ref{line:computeor:leader-q0-booltrue}: it sends a termination pulse on port~$0$
at line~\ref{line:computeor:leader-q0-send0} and then waits for a pulse on port~$1$ at
line~\ref{line:computeor:leader-q0-recv1}.

Consider any non-leader $p_i$. After Phase~1, it has $q_i = 0$, has sent
once on port~$1$, and is blocked waiting for a second message:
\begin{itemize}
  \item on any port $u_i$ at
        line~\ref{line:computeor:nonleader-false-recvu} if $bool_i =
        \texttt{false}$;
  \item on port $1-q_i = 1$ at
        line~\ref{line:computeor:nonleader-true-recv} if $bool_i =
        \texttt{true}$.
\end{itemize}
The final pulse propagates counter-clockwise: it is always sent on
port~$0$ and thus received on port~$1$ by the previous process in that
direction. Hence, for every non-leader $p_i$, the second pulse it receives
is the termination pulse on port~$1$.

	If $bool_i = \texttt{false}$, its second receive in
	line~\ref{line:computeor:nonleader-false-recvu} yields $u_i = 1$. The process
	then sends on port $1-u_i = 0$ at
	line~\ref{line:computeor:nonleader-false-send1u}, forwarding the termination
	pulse, and since $q_i = 0$ the condition $u_i = 1 \land q_i = 1$ is false; it
	therefore keeps $ans=\texttt{true}$.

If $bool_i = \texttt{true}$, the process was waiting at
line~\ref{line:computeor:nonleader-true-recv} for a pulse on port $1-q_i =
	1$. The termination pulse arrives on port~$1$ and unlock this receive; the
	process then sends on port~$0$ at
	line~\ref{line:computeor:nonleader-true-send0} and keeps $ans=\texttt{true}$.

		After the counter-clockwise termination wave returns to the leader at
		line~\ref{line:computeor:leader-q0-recv1}, the ring is quiescent and all
		processes execute the final clockwise barrier wave. Afterwards every process
		returns $\texttt{true}$ (line~\ref{line:computeor:return-true}), with the leader
		returning last.
\end{proof}

\begin{theorem}\label{thm:computeor-correct}
For any assignment of input bits $(bool_i)_{i=0}^{n-1}$, every process
terminates when executing Algorithm~\ref{algo:computeor_simplified}, and all
processes output the same Boolean: $\bigvee_{i=0}^{n-1} bool_i$.  The total number of pulses used is exactly $3n$. Moreover, 
\texttt{ComputeOR} has the composable ending property.
\end{theorem}

\begin{proof}
There are two possible values for the global OR.

\medskip\noindent
\emph{Case~1: $\bigvee_i bool_i = \texttt{false}$.}
Then all inputs are $\texttt{false}$. By Lemma~\ref{lem:all-false}, every
process terminates and outputs $\texttt{false}$, which equals the global OR.

\medskip\noindent
\emph{Case~2: $\bigvee_i bool_i = \texttt{true}$.}
Then at least one process has input $\texttt{true}$.
\begin{itemize}
  \item If $bool_\ell = \texttt{true}$, Lemma~\ref{lem:leader-true} applies,
        and every process terminates and outputs $\texttt{true}$.
  \item If $bool_\ell = \texttt{false}$, Lemma~\ref{lem:leader-false-some-true}
        applies (since there exists at least one non-leader with input
        $\texttt{true}$), and every process terminates and outputs
        $\texttt{true}$.
        \end{itemize}
In all three cases we have the composable ending by the appropriate lemma. 
       Moreover, each non-leader process sends exactly two messages in the
        OR-computation phase: one after
	its first receive (
	lines~\ref{line:computeor:nonleader-true-send1}, or
	\ref{line:computeor:nonleader-false-send1q}) and one after its second receive
	(lines~\ref{line:computeor:nonleader-true-send0}, or
	\ref{line:computeor:nonleader-false-send1u}).
	In addition, each process sends exactly one message in the final clockwise
	barrier wave.
	It is easy to observe that the leader also sends exactly two messages in the
	OR-computation phase and one message in the barrier. Hence every execution
	uses exactly $3n$ pulses.
\end{proof}

\subsection{Correctness of the \texttt{MsgExchange} Algorithm}
We are now ready to prove the correctness of the {\tt MsgExchange} algorithm. 
\begin{theorem}\label{thm:msgexchange-correct}
Consider an oriented ring with a unique leader. Let each active process execute
Algorithm~\ref{algo:all2allmsg} with input messages $msg_{cw}$ and $msg_{ccw}$
(possibly empty), and let each relay process execute
Algorithm~\ref{algo:all2allrelay}. Each
active process receives exactly the message sent to it by its clockwise and
counter-clockwise active neighbors (treating relay segments as links), and
Algorithm~\ref{algo:all2allmsg} returns these two messages (the message sent
counter-clockwise by its clockwise neighbor and the message sent clockwise by
its counter-clockwise neighbor).
Moreover, the combined execution of
Algorithms~\ref{algo:all2allmsg} and~\ref{algo:all2allrelay}  has the composable
ending property.
Let $L$ be the length of the maximum message sent.
Then the execution sends at most $15nL + 3n$ pulses. That is O(L) pulses per processes.
\end{theorem}
\begin{proof}
Collapse each maximal relay interval into a single abstract link.
By Theorem~\ref{thm:ccw-correct} and symmetry, each invocation of
\texttt{CWBitSending}/\texttt{CWBitRelay} and
\texttt{CCWBitSending}/\texttt{CCWBitRelay} terminates and each active process
receives exactly the corresponding bit from its neighbor (treating relay
segments as links). By Theorem~\ref{thm:computeor-correct}, each invocation of
\texttt{ComputeOR} terminates and returns the global OR of the local $active$
flags.
Since these subroutines have the composable ending property, repeated
applications of Theorem~\ref{thm:compose-any} imply that they can be executed
back-to-back in each iteration without cross-interference.
Therefore, in each phase, every active process sends
at most one bit clockwise and one bit counter-clockwise, and appends exactly the
bits it receives from its two neighbors. The loop continues until all active
processes have no bits left to send, which is detected simultaneously when
\texttt{ComputeOR(active)} returns \texttt{false}. Relay processes always pass
\texttt{false}, hence they terminate in the same iteration as the active ones.

Each invoked subroutine is quiescent terminating, thus the entire execution is quiescent terminating. The pulse
bound follows by summing the costs of the $L$ bit-sending phases
($\le 12nL$ pulses) and the $L+1$ calls to \texttt{ComputeOR} ($=3n(L+1)$
pulses).
Finally, since the last invoked procedure before returning from
Algorithms~\ref{algo:all2allmsg} and~\ref{algo:all2allrelay} is
\texttt{ComputeOR}, the
combined message-exchange procedure also has the composable ending property.
\end{proof}

\begin{proof}[Proof of Theorem~\ref{thm:msgxchangering}]
Fix one synchronous round of $\mathcal{A}$ on the virtual ring $R$ of active
processes (relay segments contracted into links).
In the simulation, each active process encodes the $b$ bits it would send
clockwise (resp.\ counter-clockwise) in this round as lists $msg_{cw}$ and
$msg_{ccw}$ and executes \texttt{MsgExchange}$(msg_{cw},msg_{ccw})$; each relay
process executes \texttt{MsgRelay}.
By Theorem~\ref{thm:msgexchange-correct}, each active process receives exactly
the messages sent to it by its two active neighbors in $R$, and thus it can
perform the same state update as in $\mathcal{A}$ for this round.

By the composable ending property in Theorem~\ref{thm:msgexchange-correct}, the
simulation of consecutive rounds does not create interference between pulses of
different rounds.
The complexity of $O(b)$ pulses per round follows directly from Theorem \ref{thm:msgexchange-correct}.
\end{proof}

\section{Message exchange on 2-edge-connected networks}
\label{sec:bitsendinggeneral}

In this section, we extend the message sending routine of \Cref{sec:bitsendingring} to the case of general 2-edge-connected graphs.

\thmSim*
 

We outline one step of the simulation for a 2-edge-connected network here.

    \subparagraph{Robbins cycle construction}
    
    We first use Algorithm 4 of~\cite{censor2023distributed} to obtain a Robbins cycle of the graph. A Robbins cycle on $G$ is a directed cycle that goes through all vertices of $G$, and any edge of $G$ that appears multiple times in the Robbins cycle always follows the same orientation. This is possible from the initial assumption of a unique leader and distinct IDs.

    \begin{theorem} [{\cite[Lemma 19]{censor2023distributed}}]
    Let $G = (V, E)$ be a 2-edge-connected network, and $n = |V|$. There is a content-oblivious algorithm that constructs a \emph{Robbins cycle} $C$ of length $O(n^3)$, which takes $O(n^8\log n)$ pulses, such that each process knows its clockwise and counter-clockwise neighbors for each of its occurrences in $C$.
    \end{theorem}
    
    \subparagraph{Network topology construction}
    Given the Robbins cycle, it already suffices in terms of feasibility for the purpose of simulating any algorithm.
    \begin{theorem} [Theorem 10 of~\cite{censor2023distributed}]
    Let $G = (V, E)$ be a 2-edge-connected graph, $n = |V|$, and $C$ be a Robbins cycle over $G$. Any asynchronous protocol $\pi$ that transmits $b$ bits can be simulated with $O(|C|\cdot b + |C|\log|V|)$ content-oblivious pulses.
    \end{theorem}
    However, using such a simulation naively can be costly in the long run, as each message size explodes at least by a factor of $|C|$, which is at least $\Omega(n)$. Instead, we will use the Robbins cycle only to disseminate knowledge of network topology, to devise a less costly simulation. Following the assumption of~\cite{censor2023distributed} that each ID has size $O(\log n)$, there is a protocol such that each process learns the network topology with $O(|m|\log n)$ bits (by reporting each edge), which requires $O(|m|\cdot n^3 \log n)$ pulses. Combined with step 1, $O(n^8\log n)$ pulses are required to construct the network topology at all processes.

    \subparagraph{Cycle cover finding}
    An Eulerian subgraph of $G$ is a subgraph where all vertices have even degree. We make use of the following fact.
    \begin{lemma} [{\cite[Lemma 4.1]{alon1985covering}}]
    Every 2-edge-connected graph can be covered by three Eulerian subgraphs (i.e., each edge appears in at least one Eulerian subgraph).
    \end{lemma}
    Further, every Eulerian subgraph can be broken down into cycles by iteratively removing cycles. Let every process $v$, with the knowledge of network topology, run a deterministic algorithm to obtain a cycle cover of the graph $G$ locally as $\mathcal{C} = (C_1, C_2, \ldots, C_k)$. Let $\mathcal{C}_v$ denote the subsequence (maintaining order) containing cycles in which process $v$ is present. Note that each edge appears in one, two, or three cycles.

   \subparagraph{On-cycle simulation}
    Now $v$ invoke the (efficient) $\CONGEST[b]$ simulation algorithm for rings (\Cref{algo:all2allmsg}), for each cycle contained in $\mathcal{C}_v$, in the order of $\mathcal{C}_v$. For a message that $v$ would like to send to neighbor $u$, $v$ transmits the message at its earliest possible timing, i.e., when simulating the first cycle among $\mathcal{C}_v$ that contains edge $\{u,v\}$. For subsequent cycles among $\mathcal{C}_v$ that also contains edge $\{u,v\}$, $v$ remains silent (i.e., sending empty message $\bot$) to $u$.

    Notably, when simulating message passing for one cycle, only two incident edges of $v$ are relevant. For currently irrelevant edges, $v$ holds incoming pulses while the simulation of the current cycle is on-going.

    To prove the correctness of the entire simulation, assume process $v$ is contained in distinct cycles covers $C_i$ and $C_j$. We show that a pulse sent by $v$'s neighbor $u$ while simulating $C_i$ will not be processed by $v$ if the latter is simulating $C_j$. If $u$ does not belong to $C_j$, $v$ will receive such a pulse from an irrelevant ingress port from its ongoing simulation and will hold it temporarily during the simulation of $C_j$. Now, let $u$ belong to $C_j$, hence the edge $\{u,v\}$ belongs to $C_i$ and $C_j$. It is proven in \Cref{thm:msgexchange-correct} that the simulation of $C_i$ (or $C_j$, whichever has a smaller index) has \emph{composable ending property}, hence can be safely concatenated to the latter simulation, according to \Cref{thm:compose-any}. That means both $u$ and $v$ can correctly attribute pulses to the ring simulation that spawns them, hence proving the correctness claim.

    \subparagraph{Complexity}
    
    The pulses needed at one process to simulate a round of passing a message of size $b$ in the 2-edge-connected network is $O(b)$, i.e., multiplicative overhead. This follows from the multiplicative overhead of simulating one round of $\CONGEST[b]$ on a ring as in~\Cref{thm:msgxchangering}, together with the fact that each edge is at most covered by three cycles.

\section{Computing self-decomposable aggregation functions}
\label{sec:colorcount}


In this section, we design an efficient aggregation algorithm for content-oblivious rings by combining the simulator from \Cref{thm:msgxchangering} with the \texttt{ComputeOR} algorithm. The procedure \texttt{ComputeOR} is described in \Cref{algo:computeor_simplified}, and its correctness and performance guarantees are established in \Cref{thm:computeor-correct}. The simulator is used to run an MIS algorithm over a virtual ring over the active processes. 

\begin{lemma}[Cole--Vishkin~\cite{cole1986deterministic}]\label{lem:MIS}
In the $\CONGEST[1]$ model, for a ring $C$ with a designated leader and whose processes have {distinct} $\lambda$-bit IDs, there exists a deterministic algorithm that computes an MIS of $C$ containing the leader in $O(\lambda)$ rounds.
\end{lemma}

\begin{proof}
The lemma follows from the classical Cole--Vishkin color reduction algorithm~\cite{cole1986deterministic}. In an oriented ring, a proper $k$-coloring can be transformed into a proper $O(\log k)$-coloring in one round by exchanging colors with neighboring processes. Starting from the distinct IDs as an initial proper coloring, repeating this reduction for $O(\log^* n)$ iterations yields a proper $O(1)$-coloring.

When implemented in the $\CONGEST[1]$ model, the total number of rounds required for the color-reduction process is $\lambda + O(\log \lambda) + O(\log \log \lambda) + \cdots = O(\lambda)$.

Finally, a proper $O(1)$-coloring can be converted into a maximal independent set that includes the leader in $O(1)$ additional rounds. Let $k = O(1)$ denote the number of colors. Initially, the leader joins the MIS. Then, for $i = 1, 2, \ldots, k$, each process of color $i$ joins the MIS if none of its neighbors has already joined.
\end{proof}

At a high level, our aggregation algorithm proceeds iteratively. In each iteration, starting from the original ring, we compute an MIS of the current ring by simulating the algorithm of \Cref{lem:MIS} using the simulator from \Cref{thm:msgxchangering}. Processes outside the independent set send their local aggregates to a neighboring process in the independent set and then become inactive. The remaining processes form a virtual ring on which the algorithm continues.

After $O(\log n)$ iterations, only a single active process remains, having accumulated the aggregate of all values in the ring. Each iteration incurs $O(\lambda + \beta)$ pulses per process, where the $O(\lambda)$ term arises from the MIS computation and the $O(\beta)$ term from transmitting local aggregates. Thus, the overall message complexity is $O\left((\lambda + \beta)\log n\right)$ pulses per process.

Implementing this approach involves several subtleties. In particular, both the number of processes $n$ and the maximum identifier length $\lambda$ are initially unknown. Consequently, the algorithm must incorporate mechanisms that allow all processes to detect the termination of the overall computation, as well as the completion of individual subroutines; this is achieved using \texttt{ComputeOR}.

\thmCount*
\begin{proof}
We begin with a preprocessing step that allows all processes to determine $\lambda$, which is required to run the MIS algorithm of \Cref{lem:MIS}. The preprocessing proceeds as follows. For $i = 1, 2, 3, \ldots$, we execute \texttt{ComputeOR}, where the Boolean input of each process indicates whether its identifier length is at most $i$. The loop terminates once the outcome of \texttt{ComputeOR} is \texttt{false}. At that point, all processes learn that $\lambda$ equals the largest index for which the outcome was \texttt{true}. This preprocessing step requires $O(n\lambda)$ pulses.

 The main algorithm operates in phases. At the beginning of phase $t$, there is a set $A(t)$ of \emph{active} processes, which always includes the leader, and each active process $p$ stores a value $w_p(t)$. Initially, $A(0)$ consists of all processes in the ring, and $w_p(0) = f(x_p)$ for every process $p$. Each phase consists of the following three steps.

\begin{description}
  \item[1. Loneliness test.]
  We run \texttt{ComputeOR}, where the Boolean input of each process indicates whether it is a non-leader in $A(t)$. If the outcome is \texttt{false}, then all processes know that $|A(t)| = 1$, and the algorithm terminates. This step requires $O(n)$ pulses.

  \item[2. MIS on active processes.]
  We execute the MIS algorithm from \Cref{lem:MIS} on the logical ring induced by $A(t)$ using the simulator of \Cref{thm:msgxchangering}. Let $I(t) \subseteq A(t)$ denote the resulting MIS, which, by \Cref{lem:MIS}, is guaranteed to include the leader. Each active process learns whether it belongs to $I(t)$. This step requires $O(\lambda n)$ pulses.

  \item[3. Aggregation.]
Each process in $I(t)$ becomes the center of a cluster and remains active in the next phase; thus, $A(t+1) = I(t)$. Every process $p \in A(t) \setminus I(t)$ selects a neighboring process $p' \in I(t)$, joins the cluster of $p'$ as a leaf, and sends its value $w_p(t)$ to $p'$. Each cluster center $p'$ then applies the associative operator $\oplus$ to its own value $w_{p'}(t)$ and to all values received from its leaves, and sets $w_{p'}(t+1)$ to the resulting aggregate. Recall that $\oplus$ satisfies, for any disjoint multisets $X_1$ and $X_2$,
\[
f(X_1 \uplus X_2) = f(X_1) \oplus f(X_2).
\]
This aggregation step requires $O(\beta)$ rounds in the $\CONGEST[1]$ model and can therefore be implemented using the simulator from \Cref{thm:msgxchangering}. We emphasize that the simulator works even when messages have varying lengths and the parameter $\beta$ is unknown. This step requires $O(\beta n)$ pulses.
\end{description}

 By construction, the algorithm eventually terminates with the leader as the only active process, holding the value $f\left(\biguplus_{v \in C} \{x_v\}\right)$. Since an MIS of a ring contains at most half of the processes, the number of phases is $O(\log n)$. The total message complexity is therefore
\[
O(\log n) \cdot \left(O(n) + O(\lambda n) + O(\beta n)\right)
= O\!\left((\lambda + \beta)n \log n\right),
\]
as required.

 It remains to disseminate the final aggregate to all processes. One option is to use \Cref{algo:rcvn,algo:sndn}: the leader runs \texttt{SndCount}$\left(f\left(\biguplus_{v \in C} \{x_v\}\right)\right)$ while all other processes run \texttt{RcvCount}(). The dissemination step requires $O(n\beta)$ pulses.

Alternatively, the final value can be broadcast along the hierarchical decomposition constructed by the algorithm. We traverse the phases in reverse order, and in each phase $t$, the center of each cluster broadcasts the value $f\left(\biguplus_{v \in C} \{x_v\}\right)$ to the leaves of its cluster. As in the aggregation step, this broadcast can be implemented using $O(\beta n)$ pulses. Since there are $O(\log n)$ phases, the total number of pulses required for this dissemination step is $O(\beta n \log n)$.
\end{proof}

\section{Minimum finding and multiset computation}\label{sec:minfinding}

In this section, we describe an algorithm to compute the minimum input $x_{\min}$ of the processes and then show how to use it to compute the entire multiset of inputs. Applying the algorithm of Section~\ref{sec:colorcount} to the minimum-finding problem would yield an $O(B n\log n)$-pulse algorithm where $B$ is an upper bound on the size of the identifiers and on the size of the inputs. Moreover, this algorithm would require unique identifiers. In \Cref{sec:mf}, we present an algorithm that uses $O(|x_{\min}|n)$ pulses and works without IDs, establishing the following theorem.  

\minFinding*

If we want to compute the set of inputs, we can iteratively apply the algorithm of Theorem~\ref{th:minFinding} in order to discover all the initial inputs in increasing order. If we want to compute the multiset $\{\{x_p\}\}$ of values, i.e., the set of initial values together with their multiplicities, we show in \Cref{sec:multi} we can do it by iteratively applying the algorithm of Theorem~\ref{th:minFinding} and the counting algorithm of Theorem~\ref{thm:countnsqrtn} or Corollary~\ref{cor:counting}, proving the following theorem. 

\multisetComputation*

\subsection{Minimum finding}\label{sec:mf}
In this section we prove Theorem \ref{th:minFinding}. 
Suppose that each process $p_j$ starts with an input $x_j$.
We now discuss a leader-based algorithm to find the minimum input $x_{\min}$
with message complexity $O(n\cdot B)$. This procedure will then be used to compute the multiset of the inputs. 


We first compute the length $B = |x_{\min}|$ of the minimum input $x_{\min}$ using \texttt{ComputeOR}. Starting with $B = 0$, we increment $B$ until there is an active process with an input of length $B$, and then we keep as active the processes that have an input of minimum length. This can be done using $B$ iterations of \texttt{ComputeOR} that use $O(nB)$ pulses in total\footnote{Using binary search, one can even compute the length of the minimum input with $O(n \log B)$ pulses, but since the second phase of the algorithm uses $O(nB)$ pulses, this does not improve the asymptotic complexity of the algorithm}. 
From now on, we assume that every input is represented as a bit string of equal length $B$.

The minimum is discovered bit by bit by comparing the array {\tt bits}, which contains the bit encoding of the process’s input, with the MSB at position 0 and the LSB at the last position.

The procedure is in Algorithm \ref{algo:minfinding}. Each process maintains a local boolean variable \texttt{active} and a local array \texttt{minBits}. Intuitively, \texttt{minBits} stores the prefix of the minimum input that has been discovered so far, and a process is \emph{active} exactly when its own input shares that prefix; only active processes can still be candidates for the global minimum. Initially, the \texttt{active} flag records whether the process participates in the search (the designated starting active set). When we want to find the minimum all processes are  \texttt{active}.

In phase $i$, a process $p$ reads its local bit $\texttt{bit} = \texttt{bits}[i]$. To determine the minimum bit at this position, each process computes $\texttt{hasZero} = (\texttt{active} \land (\texttt{bit} = 0))$,
and all processes invoke \texttt{ComputeOR} on this value. If the result \texttt{existsZero} is \emph{true}, then there exists at least one active process with bit $0$, so the minimum bit at position $i$ is $0$. In this case, every active process with bit $1$ sets \texttt{active}~$\leftarrow$~\texttt{false}, and all processes append $0$ to \texttt{minBits}. If instead \texttt{existsZero} is \emph{false}, then no active process has bit $0$, and all active bits must be $1$; the minimum bit is $1$, the active set remains unchanged, and all processes append $1$ to \texttt{minBits}. After this step, the set of active processes is precisely the subset of participants whose input agrees with the newly extended prefix \texttt{minBits}.

 After $\len(\texttt{bits})$ phases, the array \texttt{minBits} records the entire bitstring of the global minimum input, and each process converts it into an integer and outputs $in_{\min}$. We omit the discussion on the correctness of the algorithm described as it is immediate from the correctness and composability of the {\tt ComputeOR} procedure, the cost of the procedure is a total of $6nB$ pulses as it calls $2B$ times {\tt ComputeOR}.

\begin{algorithm}[H]
  \footnotesize
  \DontPrintSemicolon
  \SetKw{Send}{send}
  \SetKw{Receive}{receive}
  \SetKw{Return}{return}

  \caption{MinFinding(bits, active) computes the minimum input among the processes that start with active=true; at the end the active flag remains true exactly at those minima. The MSB of bits is in position 0 and the LSB is in the last position.}
  \label{algo:minfinding}

  $B \leftarrow 0$\;   
  \Repeat{$\text{\texttt{ComputeOR}}(active \land (B \geq \len(\texttt{bits})))$}{
    $B \leftarrow B +1$
  }
  $active \leftarrow active \land (\len(\texttt{bits}) = B)$\;
  $phase \leftarrow 0$\;
  $minBits \leftarrow [\,]$\;

  \While{$phase < \len(\texttt{bits})$}{
    $bit \leftarrow \texttt{bits}[phase]$\;

    $hasZero \leftarrow (active \land (bit = 0))$\;
    $existsZero \leftarrow$ \texttt{ComputeOR}($hasZero$)\;

    \uIf{$existsZero$}{
      \If{$active \land (bit = 1)$}{
        $active \leftarrow false$\;
      }
      $minBits$.append($0$)\;
    }
    \Else{
      \tcp{no active bit 0: all active bits are 1}
      $minBits$.append($1$)\;
    }

    $phase \leftarrow phase + 1$\;
  }

  $minInput \leftarrow \texttt{convertBitsToInt}(minBits)$\;
  \Return $minInput$\;
\end{algorithm}

\subsection{Multiset computation}\label{sec:multi}
In this section we prove Theorem \ref{th:multisetComputation}. 
Suppose that each process $p_j$ starts with an input $x_j \in Input$. Our goal now is to compute the multiset of inputs. We assume the existence of a leader $p_\ell$ and work in the
ID-equipped setting of Section~\ref{sec:colorcount} so that we can invoke its
counting procedure.

We maintain a local boolean flag \texttt{eligible}, initially set to
\texttt{true} at every process, indicating that its input has not yet been
accounted for.
As long as \texttt{ComputeOR}(\texttt{eligible}) returns \texttt{true}, all
processes execute \texttt{MinFinding}(\texttt{bits}, \texttt{eligible})
(Algorithm~\ref{algo:minfinding} described in the previous section).
This a algorithm identifies the current minimum
value $x_{\min}$ among eligible processes and sets $active=\texttt{true}$ at
exactly the eligible processes with input $x_{\min}$.

Next, these active processes execute the counting algorithm of
Section~\ref{sec:colorcount} on the virtual ring induced by the active set (all
other processes act as relays), and all processes learn the multiplicity
$\#(x_{\min})$.
Finally, each process updates $\texttt{eligible} \leftarrow \texttt{eligible} \land \neg
active$, removing all occurrences of $x_{\min}$ from the eligible set, and the
loop continues.
The procedure above is correct as all the algorithms composing it have the composable ending. 

\textbf{Complexity:}
Let $B$ be the input bit length, let $D$ be the number of distinct input values,
and let $k_1,\ldots,k_D$ be their multiplicities (so $\sum_r k_r = n$).
Each \texttt{MinFinding} invocation costs at most $O(nB)$ pulses, and it is executed once per distinct
value, so the total cost of minimum-finding is $O(nBD)$.
For a value of multiplicity $k_r$, the counting step uses the algorithm of
Section~\ref{sec:colorcount} to compute the sum of the local indicator bits
$active$, which costs $O(n\log^2 n)$ pulses.
Hence the total counting cost is $O(nD\log^2 n)$.
Overall, the multiset computation uses $O(nBD + nD\log^2 n) = O(nD(B+\log^2 n))$
pulses.

\section{Lower bound} \label{sec:lb}
 
In this section we prove that when processes do not know any upper bound on ring size, any algorithm that solves counting  must send $\Omega(n\log n)$ messages in some execution even if processes have distinct identifiers and a leader process $p_{\ell}$ is present. 

\subparagraph{Interface traces of intervals.}
For an interval (contiguous subpath) $I=(p_a,p_{a+1},\dots,p_b)$ we denote its two boundary edges by
\[
e_L(I)=\{p_{a-1},p_a\}, \qquad e_R(I)=\{p_b,p_{b+1}\}.
\]
Fix an execution $E$ and fix its global interleaving order of atomic events (a total order chosen by the adversary).
Let $\Sigma=\{\Liin,\Liout,\Riin,\Riout\}$.
Consider the subsequence $(d_1,\dots,d_m)$ of all \emph{delivery events} in $E$ that occur on the two boundary edges $e_L(I)$ and $e_R(I)$, listed in that global order.
Define a mapping $\pi_I$ from these delivery events to $\Sigma$ by
\[
\pi_I(d)=
\begin{cases}
\Liin & \text{if $d$ is a delivery \emph{at} $p_a$ of a pulse that arrived \emph{from} $p_{a-1}$},\\
\Liout & \text{if $d$ is a delivery at $p_{a-1}$ of a pulse that arrived from $p_a$},\\
\Riin & \text{if $d$ is a delivery at $p_b$ of a pulse that arrived from $p_{b+1}$},\\
\Riout & \text{if $d$ is a delivery at $p_{b+1}$ of a pulse that arrived from $p_b$}.
\end{cases}
\]
The \emph{interface trace} of $I$ in $E$ is the word
\[
\trace_E(I) \;:=\; \pi_I(d_1)\,\pi_I(d_2)\,\cdots\,\pi_I(d_m)\ \in \Sigma^{*}.
\]
Let $|\trace_E(I)|$ denote its length. If $M(e)$ is the number of pulses that traverse edge $e$, then
\begin{equation}
\label{eq:trace-length-bound}
|\trace_E(I)| = M(e_L(I)) + M(e_R(I)).
\end{equation}
Two intervals are \emph{trace-equivalent} if their interface traces are identical words.

\begin{observation}[Counting traces]\label{obs:trace-count}
For $k\ge 0$, the number of possible words over $\Sigma=\{\Liin,\Liout,\Riin,\Riout\}$ of length at most $k$ is
\(\sum_{i=0}^{k}4^i \le 4^{k+1}\).
\end{observation}

\begin{lemma}[Splicing with a shared boundary]\label{lem:splice-prefix}
Let $E$ be an execution of an algorithm $\mathcal A$ on an oriented ring $R$.
Let $I_1\subset I_2$ be intervals that share their left boundary,
$e_L(I_1)=e_L(I_2)=:e_L$, and assume the leader $p_{\ell}$ is not in $I_2$.
If $I_1$ and $I_2$ are trace equivalent, that is $\trace_{E}(I_1)=\trace_{E}(I_2)$, then in the ring $R'$ obtained by replacing $I_2$ with $I_1$
(preserving orientation and wiring to the same two external neighbors)
there exists an execution $E'$ in which every process outside $I_2$ has exactly the same local history as in $E$.
In particular, the leader has the same history and output in $E$ and $E'$.
\end{lemma}

\begin{proof}
Fix the global interleaving (total order) of events in $E$.
Let $O=(o_0,o_1,\dots,o_{f})$ be the subsequence of all events that are \emph{not} boundary \emph{deliveries} on $I_2$’s two boundary edges (so $O$ includes all sends and all deliveries away from those two edges). 
For each boundary delivery $\delta$ on $I_2$’s boundaries, let $i(\delta)\in\{-1,0,\dots,s,s+1\}$ be the unique index such that $\delta$ occurs strictly after $o_{i(\delta)}$ (if $i(\delta)>-1$) and  before $o_{i(\delta)+1}$ (if $i(\delta)<s+1$) in $E$ (multiple boundary deliveries may share the same index; their internal order is the one from $E$). 

Since $\trace_{E}(I_1)=\trace_{E}(I_2)$, the ordered list of left-boundary labels (those in $\{\Liin,\Liout\}$) and the ordered list of right-boundary labels (those in $\{\Riin,\Riout\}$) are identical for $I_1$ and $I_2$.
We now first show the sequence of events in $E'$ and then we show that such sequences is realizable. 

\emph{Execution $E'$ on $R'$:}
(i) replay the entire outside sequence $O$ verbatim;
(ii) for every boundary delivery $\delta$ on $I_2$ in $E$, insert in $E'$ a boundary delivery $\delta'$ placed in the same gap after $o_{i(\delta)}$ and before $o_{i(\delta)+1}$, with the same label $\pi_{I_2}(\delta)\in\{\Liin,\Liout,\Riin,\Riout\}$.
If the label is in $\{\Liin,\Liout\}$, perform $\delta'$ on the \emph{same physical edge} $e_L$; if it is in $\{\Riin,\Riout\}$, perform $\delta'$ on the right boundary of $I_1$ in $R'$ (which has the same external right neighbor as $I_2$ had).

\emph{Realizability of $E'$:}
The proof proceeds by a simple inductive argument. All events in $O$ can be replayed without problem until the first boundary event $\delta'$. 
If $\delta'$ is incoming to $I_1$ ($\Liin$ or $\Riin$), the sender is the external neighbor; since $O$ is unchanged, that neighbor issues the corresponding send in $E'$ as in $E$, and we choose a finite delay so the delivery occurs after the prescribed index in $O$ (more precisely between $o_{i(\delta)}$ and $o_{i(\delta)+1}$.).
If $\delta'$ is outgoing from $I_1$ ($\Liout$ or $\Riout$), in $E$ the processes of $I_1$ executed some internal schedule producing exactly this boundary word; by inductive hypothesis we can replicate that internal schedule inside $I_1$ in $E'$ and place the corresponding send before the prescribed delivery, choosing a finite delay so it arrives at the prescribed index in $O$ (more precisely between $o_{i(\delta)}$ and $o_{i(\delta)+1}$.).
Pulses carry no content and delays are finite but otherwise unconstrained, so this is always feasible irrespectively of the IDs of processes in $I_1$.

Because $O$ is replayed verbatim and each external neighbor sees the same ordered boundary deliveries at the same places among outside events, every process outside $I_2$, in particular the leader, has the same local history in $E'$ as in $E$.
\end{proof}

\medskip
We now prove the lower bound.

\thmlb*

\begin{proof}
Assume for contradiction that there exists a correct algorithm $\mathcal A$ and an infinite set of sizes for which the worst-case message complexity is $o(n\log n)$.
Takes one of such $n>100$, specifically one where on a ring $R$ of size $n$ there is an execution $E$ of ${\cal A}$ with total message complexity \(M\) such that: $$4^{\frac{6M}{n}+1}<n/4-1,$$ 
notice that since the message complexity of ${\cal A}$ is $o(n\log n)$ such situation exists for a sufficiently large $n$.
The proof now proceeds by showing that such execution $E$ makes the leader output a wrong count.

Recall that \(M(e)\) is the number of pulses that traverse edge $e$ (both directions) in execution $E$, so $M=\sum_{e\in Edges(R)} M(e)$ where $Edges(R)$ is the set of ring edges, \(|Edges(R)|=n\). Let \(\overline m := \frac{M}{n}\).

We now show that we can choose a sparse edge at at least distance $n/2+1$ from the leader.

By averaging and the fact that each $M(e)$ is non-negative, at least $n/2$ edges satisfy \(M(e)\le 2\overline m\), we define such edges as  sparse edges.
Among those, pick an edge \(e_l=\{u,v\}\) with both $u\neq \ell$ and $v \neq \ell$,  that maximizes the distance from one of its endpoints to the leader $p_{\ell}$; let $v$ be that endpoint, and let
\(d\) be the number of edges encountered when moving on the maximum distance arc from $v$ up to $p_{\ell}$, without loss of generality we assume that the arc is a clockwise arc, otherwise the proof is the same by just assuming the arc counter-clockwise. 
By the fact that at least $n/2>50\geq5$ edges are sparse and among any five edges, at least one has an endpoint $v$ whose longer  arc to $p_{\ell}$ has length \(d\ge n/2+1\).
Recall that we have \(M(e_l)\le 2\overline m\).

We now show that there are at least $n/4$ intervals with sparse interface traces.

For each \(t=0,1,\dots,d-1\), define the clockwise interval
$I_t := (v, v{+}1, \dots, v{+}t)$, which does not contain~$p_{\ell}$.

Then \(e_L(I_t)=e_l\), while \(e_R(I_t)=\{v{+}t, v{+}t{+}1\}\).
The edges \(e_R(I_t)\) are distinct for different \(t\) and lie on the clockwise arc from \(v\) to \(\ell\).
Thus we have:
\[
\sum_{t=0}^{d-1} M(e_R(I_t)) \;\le\; M.
\]

Averaging over all $d$ edges we have that the average number of messages exchanged on edges in $I_{d-1}$ is: 
\[
\frac{1}{d}\sum_{t=0}^{d-1} M(e_R(I_t)) \;\le\; \frac{M}{d} \;\le\; \frac{M}{n/2} \;=\; 2\overline m.
\]

Hence, using the same averaging algorithm above, at least \(d/2 \geq n/4\) indices \(t\) satisfy \(M(e_R(I_t))\le 4\overline m\).
For each such \(t\),
\[
|\trace_{E}(I_t)|
\ \le M(e_L(I_t)) + M(e_R(I_t))
\;\le\; 2\overline m + 4\overline m
\;=\; 6\,\overline m.
\]

Let \(T := \{\, t \in \{1,\dots,d-1\}: M(e_R(I_t))\le 4\overline m \,\}\).
We have \(|T|\geq n/4-1\), and the family \(\{I_t : t\in T\}\) is totally ordered by inclusion.
By Observation~\ref{obs:trace-count}, the number of distinct traces over \(\Sigma\) of length at most \(6\overline m\) is at most \(4^{6\overline m+1}\), and by our choice of the ring size and execution we have \(4^{6\overline m+1} < n/4-1\).
Hence there exist \(t_1<t_2\) in \(T\), such that
$\trace_{E}(I_{t_1})=\trace_{E}(I_{t_2})$. Set \(I_1:=I_{t_1}\) and \(I_2:=I_{t_2}\); then \(I_1\subset I_2\), they share the left border and neither contains~$p_{\ell}$. So we can apply Lemma \ref{lem:splice-prefix}.

Let \(R'\) be the ring obtained by replacing \(I_2\) with \(I_1\) (preserving orientation), and let \(n' := |R'| < n\).
By Lemma~\ref{lem:splice-prefix}, there is an execution \(E'\) of $\mathcal A$ on \(R'\) in which every process outside \(I_2\), in particular the leader \(\ell\), has the same local history as in \(E\), and therefore the same output.

Since processes do \emph{not} know \(n\), correctness of counting requires that, on a ring of size \(n\), the leader eventually outputs \(n\), and on a ring of size \(n'\), the leader eventually outputs \(n'\).
But in \(E\) and \(E'\) the leader's local history is identical, so its output is the same in both executions.
This cannot equal both \(n\) and \(n'\), a contradiction.\end{proof}

\section{Conclusions and open problems}\label{sec:conclusion}

In this work, we demonstrated that content-oblivious communication is far more powerful than previously understood. We showed that message-passing algorithms on rings, and more generally on 2-edge-connected networks, can be simulated in the content-oblivious setting with only a \emph{constant} multiplicative overhead. This constitutes a substantial improvement over the previously known simulator of Censor-Hillel~et~al.~\cite{censor2023distributed}, which incurs a large \emph{polynomial} multiplicative overhead.

On rings, our simulator operates with no preprocessing beyond the existence of a pre-selected leader. For general 2-edge-connected networks, the simulation relies on a preprocessing phase that requires sending a polynomial number of pulses. While our results establish that such preprocessing suffices to enable constant-overhead simulation, an important open question is whether the cost of this preprocessing can be significantly reduced.

A central application of our simulator is a deterministic $O(n \log^2 n)$-message counting algorithm for rings with unique $O(\log n)$-bit IDs and a leader, which nearly matches our $\Omega(n \log n)$ lower bound. Closing the gap between these bounds remains an intriguing open problem.

The picture is far less complete in the \emph{anonymous} setting, where the power and limitations of content-oblivious computation remain largely unknown. Here, despite the absence of IDs and the inherent difficulty of symmetry breaking, we showed that counting is possible using $O(n^{1.5})$ pulses in the presence of a leader. Our algorithm simultaneously solves the naming problem by assigning unique $O(\log n)$-bit IDs to all processes. Is it possible to obtain substantially more efficient solutions for counting and naming in anonymous content-oblivious rings?


 

\bibliographystyle{plainurl}
\bibliography{refs}


\appendix

\end{document}